\documentclass[11pt]{article}

\pdfoutput=1

\usepackage{graphicx}
\usepackage[misc,geometry]{ifsym}
\usepackage{verbatim}
\usepackage{amsfonts}
\usepackage{caption}
\usepackage{tabularx}
\usepackage{array}
\usepackage{booktabs}
\usepackage{multirow}
\usepackage{fancyvrb}
\usepackage{amsmath}
\usepackage{algorithm2e}
\usepackage{hyperref}
\usepackage{framed}

\usepackage{enumerate}

\usepackage{latexsym}
\usepackage{amsmath,amssymb}

\usepackage{geometry}
\def\a{\alpha} \def\b{\beta} \def\d{\delta} 
    \def\g{\gamma}
\def\G{\Gamma}  
\def\ep{\epsilon}

\newcommand{\Let}{\leftarrow}
\newcommand{\fail}{\mbox{\it fail}}

\newcommand{\notoken}{\emptyset}
\newcommand{\true}{\mbox{\it true}}
\newcommand{\false}{\mbox{\it false}}
\newcommand{\vote}{\mbox{\it vote}}
\newcommand{\doubled}{\mbox{\it doubled}}
\newcommand{\Time}{\mbox{\it time}}

\newcommand{\done}{\mbox{\it done}}
\newcommand{\phase}{\mbox{\it phase}}
\newcommand{\step}{\mbox{\it step}}
\newcommand{\Next}{\mbox{\it next}}
\newcommand{\consistent}{\mbox{\it Consistent}}
\newcommand{\polylog}{\mbox{\rm polylog}\,}
\newcommand{\Ph}{{P}}
\newcommand{\It}{{\tau}}

\def\E{\mathbf{E}}

\def\Pr{\mbox{{\bf Pr}}}
\def\whp{{w.h.p.}}
\def\Whp{{W.h.p.}}

\def\Exp{\E}

\newcommand{\beq}[1]{\begin{equation}\label{#1}}
\newcommand{\eeq}{\end{equation}}

\newenvironment{proof}{\trivlist\item[]\emph{Proof}.}%
{\unskip\nobreak\hskip 1em plus 1fil\nobreak$\Box$
\parfillskip=0pt%
\endtrivlist}

{\unskip\nobreak\hskip 1em plus 1fil\nobreak$\Box$
\parfillskip=0pt%
\endtrivlist}
{\unskip\nobreak\hskip 2em plus 1fil\nobreak$\Box$
\parfillskip=0pt%
\endtrivlist}
{\unskip\nobreak\hskip 2em plus 1fil\nobreak$\Box$
\parfillskip=0pt%
\endtrivlist}

\newcommand{\ignore}[1]{}

\newtheorem{theorem}{Theorem}
\newtheorem{lemma}{Lemma}
\newtheorem{corollary}{Corollary}

\newcommand{\etal}{{\it et al.}}

\author{
Andreas Bilke\thanks{Department of Computer Sciences, University of Salzburg, Austria.
{\tt abilke@cosy.sbg.ac.at}}
\and Colin Cooper\thanks{Department of Informatics, King's College London, UK.
{\tt colin.cooper@kcl.ac.uk}}
\and Robert Els\"asser\thanks{Department of Computer Sciences, University of Salzburg, Austria.
{\tt elsa@cosy.sbg.ac.at}}
\and Tomasz Radzik\thanks{Department of Informatics, King's College London, UK.
{\tt tomasz.radzik@kcl.ac.uk}}
}

\title{Population protocols for leader election and exact majority with $O(\log^2 n)$ states
and $O(\log^2 n)$ convergence time\thanks{%
Work supported in part by 
EPSRC grant EP/M005038/1, ``Randomized algorithms for computer networks''.
}
}

\date{February 15, 2017}

\begin{document}

\maketitle

\begin{abstract}
We consider the model of population protocols, which can be viewed as a sequence of random pairwise interactions of $n$ agents (nodes).
During each interaction, two agents $v$ and $w$, selected uniformly at random, update their states on the basis of their current states, 
and the whole system should in long run converge towards a desired global final state.
The main question about any given population protocol is whether it converges to the final state and if so, what is the convergence rate.

In this paper we show population protocols for two problems: the leader election and the exact majority voting. 
The leader election starts with all agents in the same initial state and the goal is to converge to 
the (global) state when exactly one agent is in a distinct state $L$.
The exact majority voting starts with each agent in one of the two distinct states $A$ or $B$ and the goal is
to make all nodes know which of these two states was the initial majority state, even if that majority was 
just by a single vote.
 
Doty and Soloveichik [DISC 2015] showed that any population protocol for leader election 
requires expected linear (parallel) time, defined as the number of interactions divided by $n$,
if agents have only constant number of states.
Alistarh and Gelashvili [ICALP 2015]  showed a leader-election protocol which converges in $O(\log^3 n)$ time
\whp\ and in expectation
and needs $\Theta(\log^3 n)$ states per agent.
We present a protocol which elects the leader in $O(\log^2 n)$ time \whp\ and in expectation
and uses $\Theta(\log^2 n)$ states per agent.
For the exact majority voting, we show a population protocol with the same asymptotic performance: $O(\log^2 n)$ time 
and $\Theta(\log^2 n)$ states per agent.
The exact-majority protocol proposed by Alistarh~\etal~[PODC 2015] achieves expected $O(\log^2 n)$ time, 
but requires either the initial imbalance between $A$'s and $B$'s of $\Omega(n / polylog\, n)$
or $\Omega(n/\log n)$ states per agent.
(Their protocol can achieve also $O(\log n)$ time, but with $\Omega(n)$ states per agent.)
More recently, Alistarh~\etal~[SODA 2017]
showed $O(\log^2 n)$-state protocols for both problems,
with the exact majority protocol  converging
in time
$O(\log^3 n)$ \whp\ and in expectation,
and the leader election protocol converging in 
time $O(\log^{6.3} n)$ \whp\ and 
$O(\log^{5.3} n)$ in expectation.

Our leader election and exact majority protocols are based on the idea of agents counting their local interactions
and rely on the probabilistic fact that the uniform random selection would limit the divergence of the individual counts.  
\end{abstract}

%
%
%

\thispagestyle{empty}

\newpage

\setcounter{page}{1}

\section{Introduction}

We consider population protocols~\cite{DBLP:journals/dc/AngluinADFP06} 
for leader election and exact majority voting. 
A {\em population protocol\/} specifies how two {\em agents}, or {\em nodes},
change their states when they interact.
The computation of a protocol is a (perpetual) 
sequence of interactions between two nodes.
The system consists of $n$ nodes and a {\em scheduler} which keeps selecting pairs of nodes for interaction.
The objective is that the whole system eventually stabilizes in (converges to) a configuration which 
has some desired target property.
In the general case, the nodes can be connected according to a specified graph $G=(V,E)$ 
and two nodes can interact only if they are joined by an edge. 
Following the scenario considered in the majority of previous work 
on population protocols, we assume 
the complete communication graph and the random uniform
scheduler. That is, each pair of nodes has equal probability to be selected 
for interaction in the current step, independently of the previous interactions.

The {\em leader election\/} starts with all nodes 
in the same initial state and the goal is that the system converges to 
a configuration with exactly one agent in a state which indicates that this node is the leader.
The (two-opinion) {\em exact majority voting} 
starts with each node in one of two distinct states $q_A$ and $q_B$, representing 
two distinct opinions (or votes) $A$ and $B$.
Initially $a_0$ nodes hold opinion $A$  (are in the state $q_A$)
and $b_0$ nodes hold opinion $B$, and 
we assume that $a_0 \neq b_0$. 
We denote the initial imbalance between the two opinions by 
$\ep = |a_0-b_0|/n \ge 1/n$. 
The goal is that eventually all nodes
will have the opinion of the initial majority.
The {\em exact} majority voting should guarantee the correct answer even if the 
difference between $a_0$ and $b_0$ is only $1$.

Let $S$ denote the set of states of a population protocol.
While $|S|$ can grow with the size $n$ of the population,
keeping the number of states low  is one of the primary objectives 
in design of population protocols.
Let $q(v,t) \in S$ denote the state of a node $v\in V$ at step $t$
(that is, after $t$ individual interactions).
Nodes change their states in pairwise interactions according to a common deterministic transition function 
$\delta: S\times S \rightarrow \, S\times S$.  
A population protocol has also 
an {\em output function\/} $\g: S \rightarrow \, \G$,
which indicates the property to which the system should converge.
For the exact majority voting,
$\g: S \rightarrow \, \{A, B\}$, which means that 
a node $v \in V$ in state $q\in S$ assumes that 
$\g(q)$ is the majority opinion.
An exact-majority protocol should eventually reach a step $t$, such that 
for each $v\in V$, $\g(q(v,t))$ is the initial majority opinion, and maintain
this property in all subsequent steps. 
For leader election protocols, the output function is $\g: S \rightarrow \, \{L, F\}$.
A node $v$ in a state $q$ assumes that it is the leader, if $\g(q) = L$, or 
a 'follower,' if $\g(q) = F$.

We consider undirected individual communications, that is, 
the two interacting nodes are not designated as initiator and responder, so 
the  transition functions must be symmetric.
We follow the common convention of defining the (parallel) 
time as the number of steps (individual interactions)
divided by $n/2$, that is, as the average number of interactions per node. 

The {\em completion\/} 
({\em convergence}) {\em time\/}) of a protocol
is a  random variable $T_C$ denoting the (parallel) time 
when the system stabilizes on the desired property. 
We are interested in designing protocols for which $T_C$ is small \whp\footnote{%
\whp\ --  {\em with high probability} -- means with probability at least $1 - n^{-\a}$, 
where constant $\a>0$ can be made arbitrarily large 
by increasing the (constant) parameters of the process.
}
$\E(T_C)$ is finite and ideally also small.
A finite $\E(T_C)$ implies that a protocol cannot stabilize on an incorrect answer with any non-zero probability 
(even if it was exponentially small in the size of the population). 
We note that, for example, the pull voting~\cite{aldous-fill-2014,DBLP:journals/siamdm/CooperEOR13}
and the two-sample voting~\cite{DBLP:conf/wdag/CooperERRS15} 
do not have this property because of the positive probability that the minority opinion wins.
%
In addition to the analysis of  correctness of a given protocol, we also want to derive 
bounds on the completion time, \whp\ and/or in expectation.
The aim is to achieve fast completion time with small number of states.

\subsection{Previous work}

Draief and Vojnovi\'c~\cite{DBLP:conf/infocom/DraiefV10} and 
Mertzios \etal~\cite{Mertzios-etal-ICALP2014}
analysed two similar four-state exact-majority protocols.
Both protocols are based on the idea that the two opinions have 
``passive/weak'' versions $a$ and $b$ and
``dynamic/strong'' versions $A$ and~$B$. 
The strong versions of 
opinions can be viewed as tokens 
moving around the graph. 
Initially each node $v$ has a strong opinion $A$ or $B$, and during the computation
it has always one of the opinions $a$, $b$, $A$ or $B$ (so is in one of
these four states). 
The strong opinions have dual purpose. 
Firstly, two opposite strong opinions cancel each other and change into weak opinions, if they 
interact.
Such pairwise canceling ensures that the difference between the number of opinions $a/A$ and $b/B$ 
does not change throughout the computation, remaining equal to $a_0 - b_0$.
Secondly, the strong opinions which are still alive keep moving around the graph 
'converting' the passive opposite opinions.

\begin{sloppy}

Mertzios \etal~\cite{Mertzios-etal-ICALP2014} call their protocol the {\em 4-state ambassador protocol} 
(the strong opinions are ambassadors) and prove the expected convergence time 
$O(n^5)$ for any graph 
and $O((n \log n)/|a_0 - b_0|)$ for the complete graph.
Draief and Vojnovi\'c~\cite{DBLP:conf/infocom/DraiefV10} call their 4-state protocol
the {\em binary interval consensus}, viewing it as a special case of
the {\em interval consensus} protocol of
B{\'{e}}n{\'{e}}zit  \etal~\cite{DBLP:conf/icassp/BenezitTV09}, and 
analyse it in the continuous-time model.
For the complete graph
and uniform edge rates,
they show that the expected convergence time is at most
$2n(\log n + 1)/|a_0 - b_0|$.
They also derive completion time bounds for 
cycles, stars 
and Erd{\H{o}}s-R\'enyi graphs.

\end{sloppy}

The appealing aspect of the four-state exact-majority protocols is their simplicity and the constant-size local memory, but 
their convergence time is slow, if the initial imbalance is small. 
The convergence time decreases if the initial imbalance increases, so the performance would be improved,
if there was a way of boosting the initial imbalance. 
Alistarh \etal~\cite{DBLP:conf/podc/AlistarhGV15} achieved such boosting 
by 
multiplying all initial strong opinions by $r$,
where $r$ is a positive integer parameter.
The nodes keep the count of the number of strong opinions they currently hold.
When eventually all strong opinions of the initial minority are canceled, 
there are $|a_0 - b_0|r$ strong opinions of the initial majority.
This speeds up both the canceling of strong opinions and the converting of weak opinions of the initial minority.
%
%
For complete graphs, this protocol
converges in expected $O((n \log n) /(|a_0 - b_0|r))$ time and \whp\ in 
$O((n \log^2 n) /(|a_0 - b_0|r))$ time,
while the number of states is 
$r + O(\log n \log r)$.
Thus this protocol needs either $|a_0 - b_0| = \Omega(n /\polylog n)$ or 
$\Omega(n/\polylog n)$ states to achieve a $O(\polylog n)$ time.
More recently, Alistarh~\etal~\cite{DBLP:conf/soda/AlistarhAEGR17} expanded and modified this protocol,
reducing the number of states to $O(\log^2 n)$
and the converges time to $O(\log^3 n)$ \whp\ and in expectation.


A suite of polylogarithmic-time population protocols for various functions, including the exact majority,
was proposed by Angluin~\etal~\cite{AngluinAE2008fast}, but those protocols require a unique leader 
to synchronize the progress of the computation. 
Their exact-majority protocol 
runs in $O(\log^2 n)$ time \whp\ and requires only constant number of states,
but it may fail with some small probability (polynomially small, but positive) and,
more significantly, requires the leader. 
%
The protocols developed in~\cite{AngluinAE2008fast}
are based on the idea of alternating {\em cancellations\/} 
and {\em duplications}, which has the following interpretation 
within the framework of canceling strong opinions. 
The canceling stops after a round of some pre-defined number of interactions,
reducing the overall number of strong opinions to less than $cn$, for some small constant $c < 1$.
This is followed by a round of interactions which ensures, \whp, that
all remaning strong opinions duplicate. 
Each round takes $O(\log n)$ time and $O(\log n)$ repetitions of 
the double-round of cancellations and duplications increases the difference between 
the number of strong opinions $A$ and strong opinions $B$ to $\Theta(n)$.

Berenbrink~\etal~\cite{DBLP:journals/corr/BerenbrinkFKMW16} have recently 
considered population protocols for the {\em plurality consensus\/} problem, which generalizes 
the majority voting problem to $k \ge  2$ opinions. Using the methodology introduced earlier
for load balancing~\cite{DBLP:conf/focs/SauerwaldS12}, they generalized the majority protocol 
of Alistarh~\etal~\cite{DBLP:conf/podc/AlistarhGV15} in a number of ways: 
 $k\ge 2$ opinions, arbitrary graphs, and only $O(\log n)$ time \whp\ 
for complete graphs and $k = 2$.
Their protocol, however, requires a polynomial number of states and 
$\Omega(n/\polylog n)$ initial advantage of the most common opinion 
to achieve $O(\polylog n)$ time.

The most recent population protocols for leader election are due to
Alistarh and Gelashvili~\cite{DBLP:conf/icalp/AlistarhG15}, who 
showed a protocol 
with $O(\log^3 n)$ states and $O(\log^3 n)$ time (in expectation and \whp),
and Alistarh~\etal~\cite{DBLP:conf/soda/AlistarhAEGR17}, who
showed a protocol 
with $O(\log^2 n)$ states and $O(\log^{5.3} n)$ expected and $O(\log^{6.3} n)$ \whp\ convergence time.
Poly-logarithmic time requires the number of states growing with $n$.
Doty and Soloveichik~\cite{DotySoloveichik-DISC2015} showed that any leader election protocol with 
constant number of states has at least linear expected convergence time.
Subsequently, Alistarh~\etal~\cite{DBLP:conf/soda/AlistarhAEGR17} showed that any leader-election
or exact-majority protocol with $O(\log\log n)$ states has $\Omega(n/\mbox{polylog}\, n)$ expected time.
There is an obvious two-state asymmetric leader election protocol with expected $O(n)$ time:
when two leader candidates meet, one of them turns into a follower.
A symmetric polynomial-time 
protocol can be obtained by augmenting the state by one bit, which changes 
on each interaction and is used to break the symmetry when two leader candidates meet.

\subsection{Our contributions}

We present an exact-majority protocol for complete graphs which converges to the correct answer in 
$O(\log^2 n)$ time in expectation and \whp, and uses $O(\log^2 n)$ states, improving
the time achieved in~\cite{DBLP:conf/soda/AlistarhAEGR17}. Note that it is enough for 
the nodes to know a polynomial upper bound on $n$ in order to guarantee the same 
asymptotic performance of our algorithms. In the following, we assume for simplicity that 
$n$ is known to every node.
Our protocol is based on the idea of alternating cancellation and duplication of opinions
introduced in~\cite{AngluinAE2008fast}, but 
achieves all necessary synchronization without a leader.
The nodes keep track of their {\em local clocks} -- the counters of their own interactions,
which stay sufficiently closely synchronized due to the uniform random selection of interactions.
There are $O(\log n)$ cancellation/duplication phases and each phase takes 
$O(\log n)$ time, giving the total $O(\log^2 n)$ time.
The requirement for $O(\log^2 n)$ states comes from counting $O(\log^2 n)$ interactions. 
Special care has to be given to the possibility that the local clocks may diverge or that 
some nodes have reached incorrect decision regarding the majority opinion. 
In both cases
the nodes eventually realize that something has gone wrong and switch to 
the back-up four-state protocol, which guarantees that all nodes reach the correct answer
in expected polynomial time. This is similar as in~\cite{DBLP:conf/icalp/AlistarhG15}, and
since the back-up protocol is needed only with polynomially small probability, 
the overall time is still $O(\log^2 n)$, both in expectation and \whp

For the leader election problem, we present a population protocol which has the same $O(\log^2 n)$
asymptotic running time and the number of states as our exact-majority protocol.
Both our protocols are based on nodes counting their interactions. 
Availability of $O(\log^2 n)$ states means that each node can keep count of up to
$O(\log^2 n)$ interactions.
The random uniform scheduler ensures that the individual counts do not diverge too much.
More specifically, in a period of $C n \log n$ interactions, 
each node is selected on average for $2C \log n$ interactions and, by a simple application of
Chernoff's bounds, \whp\ each node is 
selected for at least $(2C - c) \log n$
and at most  $(2C + c) \log n$ interactions, where $0 < c \ll C$ are suitably large constants.
For example, for $c=9$ and 
$C=12$, the failure probability is $n^{-\alpha}$ for some $\alpha>1$.
Both our protocols use a $O(\log n)$-time asynchronous push-pull broadcast protocol.
Our leader election protocol (the more technically involved of the two) uses also 
a technique for simulating Bernoulli trails in a population protocol.
The nodes count the number of successes within a logarithmic number 
of trials, and it can be shown that with constant probability
the maximum number of successes occurs at a single node. Such a node will become 
the leader. Since with constant probability two nodes have 
the same maximum number of successes, we need a process of testing which of these
two cases has happened and restart the protocol, if the nodes realize that there 
at least two leader candidates.

The analysis of our protocols relies on the performance of the following asynchronous push-pull broadcast
in the complete graph. 
Initially one of the nodes possesses a piece of information. 
In each step, two nodes are selected uniformly at random, and if one has the message, then 
at the end of the step both of them will have it. 
\Whp\
in $O(\log n)$ (parallel) time all nodes have the message.
The proof  of this is given in Appendix and is very similar to the proof 
in~\cite{KSSV00}.

%
%

\section{Exact majority voting in $O(\log^2 n)$ time with $O(\log^2 n)$ states}


The state of each node $v$ consists of $v.\vote \in \{A, B, \emptyset\}$ --
the current opinion held by $v$,
the counter $v.\Time$ of the interactions
and
three boolean variables $v.\doubled$, $v.\done$ and $v.\fail$.
The time counter increases by 1 with each interaction of $v$ and is viewed 
as a pair of counters $v.\phase$ and $v.step$: the first one counts $\Ph = \log n + 1$ phases 
and the second one counts $\It = C\log n$ steps in the current phase,
where $C$ is a suitably large constant.\footnote{%
For simplicity of notation, we assume that $\log n$ and $C\log n$ are integers.
More generally, whenever an expression refers to an index (or a number) of phases or steps,
we assume that it has an integer value.} 
The variables $v.\doubled$, $v.\done$ and $v.\fail$, all initialized to $\false$, indicate, respectively,
whether the opinion held by $v$ has been already duplicated in the current phase,
whether node $v$ has decided that its current opinion $v.\vote$ is the majority opinion,
and whether $v$ has realized that something has gone wrong.
We say that $v$ is in a $\fail$ state, if $v.\fail = \true$, and in a $\done$ state, if $v.\done = true$ and $v.\fail = \false$. 
Otherwise $v$ is in a {\em normal state}.
%
The counter $v.\Time$ counts only to $C\log^2 n$, so the number of states is $O(\log^2 n)$.

We want the local times at the nodes to be well synchronized, so that \whp\
each node which is not at the very beginning or very end of its current phase interacts only 
with nodes which are in the same phase.
This would require the local clocks (counters) not to differ by more than $c\log n$,
for some constant $0 < c \ll C$.
However, since the number of interactions per node reaches $\Theta(\log^2 n)$, we should 
expect 
that eventually there will be nodes with clocks differing by $\omega(\log n)$.
To keep the clocks close together, if two interacting nodes are at different phases,
then the node at the lower phase jumps to the end of its phase.
The interaction of two nodes is summarised in Algorithm~\ref{Algo},
with the update of the phase and step in lines~\ref{PhaseUpdateFrom}--\ref{PhaseUpdateTo}.
This simple mechanism of adjusting the local clocks gives
the following invariant (proof in Appendix).

\begin{lemma}\label{lem:synchrinization}
With high probability,
for each $p = 0, 1, \ldots, \Ph-1$, 
there is a (global) time step $T$ such that each node is in phase $p$ 
and at step at most $c\log n$ within this phase.
\end{lemma}

Assuming that 
all nodes are in the beginning part of the same phase $p$, that is,
for each node $v$, $v.\phase = p$ and $v.\step \in [0, c\log n)$,
the computation during this phase proceeds \whp\ in the following way.
A node $v$ is in the {\em canceling stage\/}  
when its steps are in $[c\log n, (C/2 - c)\log n)$, and 
in the {\em doubling stage\/} when the steps are in $[(C/2 + 2c)\log n, (C - 5c)\log n)$.
We refer to the interval $[(C/2 - c)\log n, (C/2 + 2c)\log n)$ as the {\em middle part\/} of the phase.
At the time when the first node $u$ reaches the next phase $p+1$, all other nodes are 
at the end of phase $p$, that is, their steps are in $[(C - 5c)\log n, C\log n)$.
The push-pull broadcast started by $u$ brings all nodes up to phase $p+1$ 
and 
at some point they all are in the beginning part of this phase. 
If there is already a node in state \done\ or \fail, then 
\whp\ each node will be in one of these two states within $c n\log n$ interactions, again
by the properties of the asynchronous push-pull broadcast.

If two interacting nodes are not in states $\done$ or $\fail$, then they 
participate in the canceling and doubling stages: lines ~\ref{CancelFro}--\ref{DoubleTo}.
If both nodes are in the canceling stage, then the transition from the four-state protocol applies:
if the nodes have opposite votes, they change to no-vote $\notoken$.
If both nodes are in the doubling stage, exactly one of them has a vote and it has not been duplicated 
yet in this phase, then this vote is duplicated now 
and shared between both nodes (lines~\ref{DoubleFrom}--\ref{DoubleTo}),
and the variables $\doubled$ are set in both nodes.
If a node has 
at the end of a phase a vote not duplicated during this phase
and all nodes are in normal states, then 
\whp\ the minority opinion has been already eliminated from the system.
Such a node changes to state $\done$ (lines~\ref{resetDoubledFrom}--\ref{resetDoubledTo}).

\RestyleAlgo{boxruled}
\LinesNumbered
\SetEndCharOfAlgoLine{}

\begin{algorithm}
\lIf{$v.\fail \vee u.\fail$}{\label{FailPropagateFrom}
    $v.\fail, u.\fail \Let \true$; \label{FailPropagateTo}
}
\lElseIf{$\neg \consistent(v.\Time, u.\Time)$}{\label{BigDiscrFrom}
    $v.\fail, u.\fail \Let \true$; \label{BigDiscrTo}
}
\uElseIf{$v.\done \wedge \neg  u.\done$}{\label{DonePropagateFrom}
   \lIf{$u.\vote = \emptyset$}{
       \{ $u.vote \Let v.vote$;~ $u.\done \Let \true$; \}
   }
   \lElseIf{$u.\vote \neq v.\vote$}{
       \{ $v.\fail, u.\fail \Let \true$; \}
   }
}
\lElseIf{$u.\done \wedge \neg  v.\done$}{
   [analogous to above]
}
\uElseIf{$v.\done \wedge u.\done$}{
   \lIf{$v.\vote \neq u.\vote$}{
        $v.\fail, u.\fail \Let \true$;\label{DonePropagateTo}
   }
}
\uElseIf{both nodes in the canceling stage and have opposite votes}{\label{CancelFro}
            $u.\vote \Let \notoken$;~ $v.\vote \Let \notoken$\;\label{CancelTo}
}
\ElseIf{both nodes in the doubling stage and exactly one has a vote}{\label{DoubleFrom}
        let $v$ be the node with a vote (the other case is symmetric)\;
        \lIf{$\neg v.\doubled$}{
           \{ $u.\vote \Let v.\vote$;~
           $v.\doubled, u.\doubled \Let \true$; \}\label{DoubleTo}
        }
}
\uIf{ $v.\phase \neq u.\phase$ }{\label{PhaseUpdateFrom}
        let $x\in \{ v, u\}$ be the node at the earlier phase\;
        $x.\step \Let C\log n - 1$\;
}
\For {each $x \in \{ u,v\}$}{
     $(x.\phase, x.\step) \Let \Next (x.\phase, x.\step)$\label{PhaseUpdateTo}
         \tcp*{progress to the next step}
     \uIf{$x.\step = 0 \wedge x.\vote \neq \emptyset \wedge \neg x.\fail \wedge \neg x.\done$}{\label{resetDoubledFrom}
        \lIf{$\neg x.\doubled$}{
          $x.\done \Let \true$;
        }
        \lElseIf{$x.phase < \Ph$}{
          $x.\doubled \Let \false$;\label{resetDoubledTo}
        }
     }
}
 \caption{Exact-majority: the update of the states of two interacting nodes $v$ and $u$\label{Algo}}
\end{algorithm}

The \done\ state  propagates through the system, but may at some point change to the $\fail$ state,
if there are still two opposite votes (low but positive probability);
see lines~\ref{DonePropagateFrom}--\ref{DonePropagateTo}.
The $\fail$ state is propagated through the system and does not change to any new state;
see line~\ref{FailPropagateFrom}. 
Two nodes may also get into the $\fail$ state, if their local clocks are too far apart (line~\ref{BigDiscrFrom}). 
We say that 
the local clocks of two nodes are {\em consistent}, if they belong to the same part of the same phase or
to two consecutive parts of a phase, 
for example, if one clock  belongs to the initial part $[0, c\log n)$ 
of one phase while the other belongs to the canceling stage 
of this phase
or to the final part 
of the previous phase. 
The details of this condition are hidden in the predicate $\consistent$ in line~\ref{BigDiscrFrom}.
If the two clocks are not 
consistent, then they are considered too far apart and both nodes move to the $\fail$ state.


We now summarize the proof of correctness and performance of our algorithm.
Lemmas~\ref{Lem:cancel}, \ref{Lem:double} and~\ref{Lem:phase}
formalize what happens during one phase.
Lemma~\ref{Lem:doublingVoteCount} is the basis for the claim that the computation
completes within $O(\log (n/|a_0 - b_0|)$ phases. 
Lemma~\ref{Lem:phase} follows from Lemmas~\ref{Lem:cancel} and~\ref{Lem:double}.

\begin{lemma}\label{Lem:cancel}
If all nodes are in normal states and in the beginning part of the same phase $p$, $0 \le p < \Ph$, then \whp\ when the last 
node completes the canceling stage of this phase 
(a) no node has entered yet the doubling stage of this phase
(so all nodes are in the middle part of this phase), and
(b) the number of nodes with opinions is at most $n/4$ or no minority opinion is left in the system.
\end{lemma}

\begin{lemma}\label{Lem:double}
If all nodes are in the middle part of the same phase $p$, $0 \le p < \Ph$ (that is, between the canceling and doubling stages of this phase) 
and the number of nodes with opinions is at most $n/4$, then \whp\
each opinion will be duplicated in the doubling stage of this phase.
\end{lemma}

\begin{lemma}\label{Lem:phase}
If all nodes are in the beginning part of the same phase $p< \Ph - 1$ and  in normal states, then
\whp\ at some later step all nodes will be in the beginning part of phase $p+1$ and either
(a) all nodes will be in normal states or (b)
there will be at least one node in a "$\done$" state and holding the majority opinion,
no nodes in a "$\fail$" state,
and no minority opinion will be left in the system.
\end{lemma}

\begin{lemma}\label{Lem:doublingVoteCount}
When all nodes are in the beginning part of the same phase $p$, $0 \le p < \Ph$, and all are in normal states,
then the difference 
between the number of opinions $A$ and $B$ is equal to $2^p|a_0 - b_0|$.
\end{lemma}

\begin{theorem}\label{Thm:ourExactMajorityProtocol}
The protocol reaches the correct answer in $O(\log n \log (n/|a_0 - b_0|))$ parallel time \whp\
This protocol can be extended to a $O(\log^2 n)$-state population protocol for the exact-majority problem 
which converges to the correct answer in $O(\log^2 n)$ time \whp\ and in expectation.
\end{theorem}

%
%
\section{Leader election in $O(\log^2 n)$ time with $O(\log^2 n)$ states}

\label{sec:alg}

Our leader election protocol consists of two parts.
The first part is an $O(\log n)$-time process of selecting leader candidates.
\Whp\ all nodes perform the same $\Theta(\log n)$ number of Bernoulli trials and
the leader candidates
are the nodes which collect the maximum number of successes.
At least one candidate is selected and exactly 
one candidate is selected with constant probability.\footnote{%
{\em Constant probability\/} 
means in this paper probability $p$ for a constant $0 < p < 1$.} 

The second part of the protocol tests whether there is exactly one candidate.
There are 
$O(\log n)$ {\em test phases} and each phase takes \whp\ $O(\log n)$ time
and gives either negative or non-negative answer.
If there is exactly one candidate, then each test phase gives the non-negative answer.
If there are two or more candidates, then each phase gives 
the negative answer with constant probability,
so the first negative answer is obtained \whp\ within $O(\log n)$ phases.
The first phase with the  negative answer causes the restart of 
the whole computation: the push-pull broadcast moves all nodes back to
the beginning of the first part.
When $\Theta(\log n)$ consecutive phases give only non-negative answers, 
then \whp\  there is exactly one candidate, so
each candidate
declares itself the leader. 
In the low-probability event that two or more 
candidates declare themselves leaders, 
direct duels between leaders leave eventually only one surviving leader.

Since the first part of (each restart of) 
our protocol selects a unique candidate with constant probability,
the number of restarts is constant in expectation and $O(\log n)$ \whp\ 
The number of test phases between two consecutive restarts can be bounded
by a geometric variable 
with constant expected value. 
These geometric variables are independent, so the total number of test phases 
(over all restarts) is \whp\ $O(\log n)$ and the total time is \whp\ $O(\log^2 n)$.

Each node can be in one of $\Theta(\log^2 n)$ states denoted 
by $<q_i,X,G>$, where $q_i$ indicates the phase 
or role of the node, 
$X$ is the set to which the node belongs, which is 
$A$ or $B$, or $N$ for ``not defined'', and $G \in \{ N, 0, 1\}$
supports Bernoulli trials as explained below.
The initial value for $X$ and $G$ is $N$.
The ``$q_i$'' in the state notation can come with parameters
presented as superscripts; for example, $q_0^i$ or $q_4^{s,t}$.
In the description and analysis of our protocol, we may drop some part 
of the state notation.
For example, a state $q_0^i$ will mean 
a state $<q_0^{i,},X,G>$ for any values of $X$ and $G$, 
and a state $q_0$ will mean a state $q_0^i$ for any value of $i$.
The terms like ``a $q_i$ node'' and ``an $A$ node" will mean 
a node in a $q_i$ state and a node in set $A$ (in a state with $X = A$), respectively.

\subsection{Selecting leader candidates}
\label{elect}

The first part of the protocol can be viewed as consisting of five phases: 
(i) decide the $X$ value; (ii) decide the $G$ value; (iii) wait to ensure that \whp\ all nodes have decided both their 
$X$ and $G$ values; (iv) perform Bernoulli trials; 
(v) eliminate from the leader contention all nodes with the number of successes less
than the maximum.
Let $c \ll C$ be two suitably large constants. 
The states are:
$q_0^i$ with $0 \leq i \leq c \log n$;
$q_1^i$ with $1 \leq i \leq (C/4) \log n$; 
$q_2^i$ with $0 \leq i \leq C \log n$;
$q_3^{s,i}$ and $q_4^{s,t}$ with
$0 \leq i \leq C^2 \log n$ and 
$0 \leq s,t \leq C \log n$; and
$q_5^{s,i}$ with $0 \le s \leq 2C \log n$ and $0 \le i \leq 3C \log n$. 
Index $i$ counts the interactions,
$t$ counts the Bernoulli trials,
and $s$ is related to the number of successes.
The first part of the protocol ends when the first node enters 
state $q_6^{0,0,0}$. The overview of the 
computation is below and the details of the state transitions are included in Appendix.

At the beginning,
each node starts in a state $<q_0^i,N,N>$, where $0 \leq i \leq c \log n$.
We could assume that initially
all nodes are in the same state $<q_0^0,N,N>$,
but we need in the analysis this slack of $0 \leq i \leq c \log n$
to cover the restarts.
Every node, based on 
its interactions with other nodes,
joins either set $A$ or $B$ (the first phase) and sets its $G$-value to either $0$ or $1$ (the second phase). 
W.h.p.\ 
for each of the four combinations of set $A$ or $B$ and the $G$-value $0$ or $1$,
the number of nodes having this combination will be
$\Theta(n)$. 
The $B$ nodes perform $C \log n$ 
independent Bernoulli trials: each meeting with
an $A$ node is a trial,
and the success is meeting an $A$ node with the $G$ value~$1$.
%
If $B$ nodes were to test the $G$ bit of other nodes at each interaction, including the interactions 
with other $B$ nodes, then the counts of successes would not be independent. 
We show that with constant probability
the highest number 
of successes at a node is unique among the $B$-nodes.

The nodes in states $q_0^i$ have not decided yet their $X$ value.
The nodes in states $q_1^{i}$ have already decided whether their $X$ is $A$ or $B$.
If a $<q_0^{i},N,N>$ node $v$ interacts with a $<q_1^{j},A/B,N>$ node, then it 
joins the other set and switches to $<q_1^{0},B/A,N>$.
If a node $v$ reaches state $<q_0^{c \log n},N,N>$, 
that is, $c \log n$ interactions without seeing a $q_1$ node, and
meets again on its next interaction a node not in $q_1$, then $v$ decides to join set $A$,
switching to state $<q_1^{0},A,N>$.

The transitions for the $q_1$ phase convert the $G$-entry of each node from $N$
to $0$ or $1$.
A $q_1$ node first waits for $C/4 \cdot \log n$ interactions, to ensure that \whp\ all nodes 
have made their $A$-or-$B$ decisions, and then
sets its $G$-value to $0$ or $1$, if it meets a node with $X\in \{A,N\}$ or $X = B$,
respectively, and enters state $q_2^0$.
%
The $q_2$ phase is simply waiting for $C \log n$ interactions, to ensure 
that \whp\  all nodes have decided both their $X$ and $G$ values.
A node which reaches the $q_2^{C \log n}$ state, switches to state $q_4^{0,0}$ on its next interaction.
Nodes in $q_4$ states perform Bernoulli trials,
going through states $q_4^{s,t}$, where $t$ counts 
the number of interactions with $A$-nodes (the trials) and
$s$ counts the number of $1$'s seen during these interactions
(the successes). 
After $C \log n$ interactions with $A$-nodes, a $q_4^{s,C \log n}$ node
switches to $q_5^{s',0}$, where 
$s' = s$ for $A$ nodes and $s' = s + C \log n$ for $B$-nodes.
This effectively eliminates all $A$ nodes as leader contenders.
We will show that that the maximum value~$s$ occurring 
at a $<q_4^{s,C \log n}, B, *>$ node is unique with  constant probability (Lemma
\ref{lem2}).

In the next phase, 
the maximum $s$-value at $q_5^{s,i}$ nodes is broadcast to all nodes. 
If a $q_5$ node encounters a higher $s$ value, then it knows that it
is not a leader candidate,
switches to state $q_3$ and participates in broadcasting the maximum 
$s$ value.
More precisely, if a node $v$ in state 
$q_5^{s',i'}$
interacts with a node $w$ with a higher $s$ value,
that is, in a state $q_5^{s'',i''}$ or $q_3^{s'',i''}$ with $s''>s'$, 
then $v$ takes the $s$ value of $w$ and converts into a follower, switching to $q_3^{s'',0}$.
The new higher value~$s''$ is propagated further by $v$. 
If a node does not convert to $q_3$ within $3C \log n$ steps, then it assumes that its $s$ value is maximum and 
declares itself a leader candidate by switching to state $q_6^{0,0,0}$.

At least one node eventually turns into the $q_6^{0,0,0}$ state, but 
with constant probability this node is not unique because with constant probability the node achieving the highest $s$ value in 
the $q_4$ phase is not unique. 
In the second part of our protocol, described in Section \ref{sec3},
the nodes will find out whether the leader candidate is unique or not. 
If the 
leader candidate is
unique, then this node becomes the leader. 
Otherwise, a leader candidate that detects another candidate 
broadcasts this 
information through the system, restarting the whole protocol.

\subsection{Analysis of the first part of the protocol: selection of leader candidates}
\label{AnalysisPartOne}

The analysis of this algorithm is divided into three parts.
Initially each node is 
in some state $<q_0^i,N,N>$ with $0 \le i \le c \log n$. 
First we show that at the time step
when the first one or two nodes switch to state $q_4^{0,0}$,
\whp\ 
all other nodes are in states $<q_2^i,X,G>$ with $X\in \{ A,B\}$, 
$G \in \{ 0,1\}$ and $i > (C-3c)\log n$.
Additionally, \whp\
for each of the four combinations of $X$ $\in \{A,B\}$ and $G \in \{0,1\}$, 
the number of nodes in states $<q,X,G>$, 
where $q$ is $q_2^i$ or $q_4^{0,0}$, is~$\Theta(n)$.

Next, we show that \whp\ there will be a time step, in which one or two nodes are in state $q_5^{s, C \log n}$
and all other nodes are in states $q_5^{s', i'}$, with $i' < C \log n$,
and $q_3^{s'',i''}$.
Furthermore, with constant 
probability there is a unique node $w \in V$ in a state $q_5^{s_{\max},i}$, where 
$s_{\max}$ is the maximum $s$ over the nodes in states $q_5^{s,i}$, $i \leq C \log n$. 
%
Finally, there will be a time step when one node switches to state $q_6^{0,0,0}$, 
while \whp\ all other nodes are in states $q_3$ or $q_5^{s,i}$, where $i > (2C-5c) \log n$.  

\begin{lemma}
\label{lemAB}
With high probability there is a time step, at which one or two nodes are 
switching to state $q_2^{0}$ while all other nodes are in 
states $q_1^{i}$, with $i \leq C/4 \cdot \log n$. 
Thus
each node has already decided whether it is in set $A$ or $B$, and \whp\
the number of $A$ nodes and the number of $B$ nodes are both  
at least $n(1-o(1))/32$.
\end{lemma}

\begin{lemma}
\label{lem1}
When the first node switches to state $q_4^{0,0}$, then \whp\ there is no node in state $q_1$.
Moreover, \whp\ there are $\Theta(n)$ 
nodes in states $<*,x,g>$, for each $(x,g)\in \{ A,B\}\times \{ 0,1\}$. 
\end{lemma}

\begin{lemma}
\label{lemprev2}
When the first node enters state $q_5^{s,C \log n}$, \whp\ all other nodes are in states $q_5$ and $q_3$.
\end{lemma}

\begin{lemma}
\label{lem2}
Let us consider the first step, in which no node is in state $q_4$ 
(the last $q_4$ node has just switched to $q_5$). 
Then \whp\ all nodes are either 
in states $q_5^{s,i}$ with $i < C \log n$ 
or states $q_3$. Furthermore, 
for
$s_{\max} = \max \{ s ~|~ \exists \mbox{$v$ in $q_5^{s,i}$ or $q_3^{s,i}$}\}$,
and $V_{\max} = \{ v ~|~ v \mbox{ in } q_5^{s_{\max},i}\}$, 
always $|V_{\max}| \ge 1$, with constant probability 
$|V_{\max}| = 1$,
and \whp\ $|V_{\max}| = O(\log n)$.   
\end{lemma}

\begin{lemma}
\label{lem3}
Just before the step when the first node changes from $q_5^{s,3C \log n}$ to $q_6^{0,0,0}$,  
\whp\ each node is either in state $q_5^{s_{\max},i}$ (that is, belongs to $V_{\max}$)
or on state $q_3^{s_{\max},i}$.
\end{lemma}

\subsection{Testing the number of leader candidates}
\label{sec3}

The testing whether there is one or more than 
one leader candidates starts 
when the first candidate enters state $q_6^{0,0,0}$. 
The states in this part of the protocol are:
$q_6^{l,i,j}$ with $l \in \{0,1,2\}$, $0 \leq i,j \leq C \log n$;
$q_7^{l,i}$ with $0 \leq j \leq (3/4)C \log n$, if
$l \in \{ 1,2\}$, and $0 \leq j \leq C^3\log^2 n$, if $l = 0$;
and $q_8$, $q_9$, and $q_{10}$.
Each leader candidate $v$ remains in states $q_6^{l,i,j}$, 
where $j$ counts the test phases, 
$i$ counts the interactions within the current phase,
and $l$ indicates the type of message the candidate $v$ is broadcasting 
in the current phase. 
The value of $l$ equal to $1$ or $2$ indicates 
the message of type $1$ or $2$, respectively, and $l = 0$ means that $v$
has not decided yet which of the two messages to broadcast.
The general idea is that if there are two (or more) leader candidates,
then with constant probability, 
in the same test phase one candidate broadcasts message 1 
and the other message 2. 
The nodes will realize that different types of messages are in the system
(that is, that this test phase gives the negative outcome)
and the protocol will be restarted. 
The test phase gives a non-negative outcome, if there is only one leader 
candidate or if all leader candidates decide to broadcast the same message.
The nodes in states $q_7^{l,i}$ are {\em followers}, which are waiting 
for a message from a leader candidate, if $l = 0$, or are
participating in propagating message $l$, if $l \in \{1,2\}$. 
Index $i$ counts the interactions performed
since receiving the last message. 

A leader candidate starts its $j$-th test phase in state $q_6^{0,0,j}$
and decides at the first interaction which message to broadcast.
When it meets an $A$ node, then it decides 
on message $1$ and 
switches to state $q_6^{1,1,j}$, 
otherwise it decides on message $2$ and switches to state $q_6^{2,1,j}$. 
When a leader candidate has decided to send out 
in the current test phase a message of type $l\in \{1,2\}$, then 
in the next interaction, if it meets a node 
in state $q_7^{0,i}$, then this follower switches to state 
$q_7^{l,1}$. 
If the leader candidate 
interacts in this step with a node that possesses already some message (of any type), 
then it switches to state 
$q_8$, meaning that \whp\ there are at least two leader candidates,
so the whole protocol should restart. 
If this is not the case, then in the next $C \log n -1 $ steps, the leader candidate 
increments its $i$ counter in each step but does not pass on any messages 
(i.e., does not convert any other $q_7^{0,i}$ node to $q_7^{l,1}$).
The leader candidate 
switches to the restart state $q_8$, if it meets a node with a message of different type.
After $C \log n$ interactions, if the leader 
candidate has not switched to state $q_8$, it
proceeds to the next test phase by entering state $q_6^{0,0,j+1}$.
The limit on the number of test phases is set at $C \log n$. 
If a leader candidate does not recognize any other leader candidates 
during $C \log n$ phases, then it switches to state $q_9$,
declaring itself a leader.

If a follower in state $q_7^{l,i}$, with $l \in \{ 1,2\}$ (i.e., an informed  
follower propagating message $l$) 
meets a node in state $q_7^{0,i}$ 
(a follower without any message), then the message is passed to 
the uninformed follower,
who switches to $q_7^{l,1}$.
If two followers carrying different types of messages meet,
then they switch to state $q_8$, triggering a restart.
Otherwise, once the interaction counter of an informed follower reaches 
$C \log n/32$, it stops 
transmitting the message further, 
and when the counter reaches $3 C \log n/4$, it drops the message altogether and
switches to $q_7^{0,0}$ (gets ready for the next test phase).

To restart the protocol, 
when a $q_8$ node meets a node 
which is in a state other than the final leader and non-leader states $q_9$ and $q_{10}$, 
both nodes 
switch to the initial  $<q_0^0,N,N>$. If a $q_0$ node meets a node $u$ which is
in a state other than $q_0$, $q_1$, $q_9$, $q_{10}$,
then $u$ also switches to $<q_0^0,N,N>$.

In our analysis, we first show that if there is only one leader candidate, 
then this node initiates every 
$C \log n+1$ interactions a new message 1 or 2, which is then 
broadcast to all nodes in the system.
The followers 
carry out these broadcasts, but they 
forget the message after $3C \log n/4$ interactions, switching back to the 
listening mode (states $q_7^{0,i}$).
Thus
when a new message is initiated, \whp\ all 
followers are in the listening mode, so the leader 
candidate will see in the system only its current message.
If there are two leader candidates, then \whp\ one of them will realize 
that there is  another 
candidate, if
the time difference
between the initialization of their messages is 
sufficiently large, even if the messages are 
the same. 
If the leader candidates remain closely synchronized, then we rely on the fact that \whp\
in one of $\Theta(\log n)$ test phases, 
one candidate initiates message $1$ while some other candidate initiates 
message $2$. 
When this happens, \whp\ some node will realize 
that there are at least two leader candidates in the system.
In both cases, any node which realizes that \whp\ there are at least two leader candidates
switches to state $q_8$ to trigger restart.
The details of the protocol and the proofs of the following two theorems are in 
the Appendix.

\begin{theorem}
\label{theo1}
With probability 
$1-n^{-\Omega(1)}$,
the leader election algorithm designates a single leader 
in $O(\log^2 n)$ time. 
\end{theorem}

To guarantee that our algorithm 
always (that is, with probability $1$) ends in a configuration with 
exactly one node in the leader state $q_9$ and all other nodes 
in the non-leader state $q_{10}$, we proceed in the following way. 
If a $q_9$ or $q_{10}$ node meets any node in a state other than $q_9$, 
it converts this node into 
state $q_{10}$. 
Furthermore,
a $q_9$-node can be in 
two different states, $<q_9,0>$ or $<q_9,1>$, flipping the $0/1$ bit 
at each interaction. 
If a $<q_9,0>$ node meets a $<q_9,1>$ node,
then the latter switches to the non-leader state $q_{10}$. 
Such duels between leader candidates and flipping the additional bit
ensure that if two or more nodes enter the leader state $q_9$
(a low but positive probability), all but one leader will eventually switch to 
the non-leader state.
One can check that the only 
stable configurations of our protocol are when  
one node is in state $q_9$ and all other nodes are in state $q_{10}$. 
Furthermore, for each other configuration there is a sequence 
of interactions which lead to a stable configuration.
This implies that the algorithm converges with probability $1$.

We believe that 
our protocol, as described so far, has also \emph{expected} 
$O(\log^2 n)$ convergence time, but we do not have a proof.
(Convergence with probability $1$ together with fast convergence \whp\
do not necessarily imply the expected fast convergence.)
We therefore extend our protocol to get $O(\log^2 n)$ convergence time
also in expectation as stated in the theorem below.

\begin{theorem}\label{Thm:LreaderElectionExpectedLog2}
The leader election protocol can be combined 
with a polynomial-time constant-space protocol 
to obtain a leader-election population protocol with $O(\log^2 n)$ states 
and $O(\log^2 n)$ convergence time \whp\ and in expectation.
\end{theorem}

\section{Simulations}

We have implemented our exact-majority and leader election protocols
and the protocols proposed 
in~\cite{DBLP:conf/icalp/AlistarhG15}, 
\cite{DBLP:conf/soda/AlistarhAEGR17}
and \cite{Mertzios-etal-ICALP2014}.
We have run simulations on complete graphs to compare the 
observed performance of these algorithms.

The performance of the exact majority protocols is plotted 
in Figures~\ref{fig:majority-rounds} and~\ref{fig:majority-rounds-normalized}.
Figure~\ref{fig:majority-rounds} shows
that the round count for our $O(\log^2 n)$-time 
BCER protocol is quite high in comparison with
the $O(\log^3 n)$-time AAEGR protocol of~\cite{DBLP:conf/soda/AlistarhAEGR17}. 
Figure~\ref{fig:majority-rounds-normalized}, which shows normalized running times, 
suggests that both algorithms complete the computation within
$\Theta(\log^2 n)$ expected number of rounds.
The high constant factor in the running time of our algorithm 
is the price we pay for
provably guaranteeing the $O(\log^2 n)$ time
in expectation and with high probability.

We compare the performance of our $O(\log^2 n)$-states leader election protocol with 
the $O(\log^3 n)$-states protocol AG from~\cite{DBLP:conf/icalp/AlistarhG15}
and the $O(\log^2 n)$-states 
protocol AAEGR from~\cite{DBLP:conf/soda/AlistarhAEGR17}. 
The plots in 
Figures~\ref{fig:leader-rounds}
and~\ref{fig:leader-rounds-normalized} indicate that 
our protocol has the lowest rate of growth of the running time.

\newpage

\bibliographystyle{abbrv}

\newpage

\section*{Appendix}

\appendix
\setcounter{section}{1}

\subsection{Four-state exact-majority voting protocol}

The transition function of the four-state exact-majority voting protocols analysed 
in~\cite{DBLP:conf/infocom/DraiefV10} and~\cite{Mertzios-etal-ICALP2014} is geven below.
The states $X$ and $Y$ stand for distinct strong opinions and  
$x$ and $y$ stand for the corresponding weak opinions; $q$ stands for any state.

\begin{center}
\begin{tabular}{|l|l|}\hline
$(q',q'')$ & $\d(q',q'')$  \\ \hline\hline
$(X,Y)$ & $(x,y)$ in~\cite{Mertzios-etal-ICALP2014}, $(y,x)$ in~\cite{DBLP:conf/infocom/DraiefV10} \\ \hline
$(X,y)$ & $(x,X)$  \\ \hline
$(X,x)$ & $(X,x)$ \\ \hline
$(x,y)$ & $(x,y)$ in~\cite{Mertzios-etal-ICALP2014}, $(y,x)$ in~\cite{DBLP:conf/infocom/DraiefV10}\\ \hline
$(q,q)$ & $(q,q)$ \\ \hline
\end{tabular}
\end{center}

\subsection{Proofs of the lemmas for the exact majority protocol}

%

{\bf Proof of Lemma \ref{lem:synchrinization}:}

The lemma can be proven by induction. The case $p = 0$ is obvious: at the global step $T=0$, all
nodes are in phase $0$ and step $0$. 
Assume that for some $0 \le p < \Ph-1$, each node is in phase $p$ and at step at most $c\log n$.
When the first node eventually enters the next phase $p+1$ all other nodes are "dragged" 
to phase $p+1$ by the asynchronous push-pull broadcast process, 
which completes with probability at least $1-n^{-\a_1}$ 
in $(c/4)n\log n$ global steps.
During these $(c/4)n\log n$ steps, any fixed node $v$ is selected for interactions on average $(c/2)\log n$
times and no more than $c\log n$ times with probability at least $1 - n^{-\a_2}$
(from Chernoff's bounds). 
Thus, with probability at least $1 - n^{-\a_1} - n^{-\a_2} \ge 1 - n^{-\a}$, 
at the time when the last node reaches phase $p+1$, node $v$ is still at this phase and at step 
$v.\step \le c\log n$. 
This implies that the statement in the lemma holds with probability at least $1- \Ph/n^{-\a+1}$.
The constants $\a_1$, $\a_2$ and $\a$ depend on $c$ and increase with~$c$.


{\bf Proof of Lemma \ref{Lem:cancel}:}

Consider the subsequent $(C/4)n \log n$ interactions.
Each node participates in $(C/2) \log n$ interactions on average,
and \whp\ participates in at least $(C/2 - c)\log n$ and at most $(C/2 + c)\log n$ interactions
(from Chernoff's bounds).
Since each node starts with its step in  $[0, c\log n)$, then \whp\
after $(C/4)n \log n$ interactions all nodes are still in phase $p$ and their steps 
are in $[(C/2 - c)\log n, (C/2 + 2c)\log n)$, that is, all nodes are in the middle part of phase $p$.

Consider a node $v$ with the minority opinion at the beginning of phase $p$ and consider 
the interactions of this node during this phase at steps $[(3c)\log n, (C/2 - 3c)\log n]$.
For each such interaction, we say that this interaction is successful 
if the other node $u$ is in the canceling stage of this phase and either 
the total number of majority opinions is less than $n/8$, 
or this number is at least $n/8$ and node $u$ holds the majority opinion
(so the opinions of $v$ and $u$ are canceled).
Since node $u$ is \whp\ in the canceling stage of this phase, the probability of success 
is at least $1/8 -o(1)$. 
Therefore, \whp\ node $v$ will have success in at least one interaction,
so \whp\ each node with the minority opinion will have at least one success.
This implies that when all nodes reach the middle part of this phase, then \whp\
either the number of nodes with the majority opinion is less than $n/8$
(so there are less than $n/4$ nodes with opinions)
or all minority opinions have been canceled.

%

{\bf Proof of Lemma \ref{Lem:double}:}

(Sketch)
An opinion is duplicated, if it meets a node with no opinion.
If before the doubling stage the number of nodes with opinions is at most $n/4$, there will be at least 
$n/2$ nodes with no opinion throughout this doubling stage. 
Thus each node $v$ which holds an opinion, in each step $[((3/5)C)\log n, ((4/5)C)\log n)$,
has at least $1/2 - o(1)$ chance to meet a node with no opinion (and to duplicate). 

%

{\bf Proof of Lemma \ref{Lem:doublingVoteCount}:}

By induction on $p$.
For $p = 0$, the statement is obvious, since each node has its initial opinion.
Assume that the statement is true for some $0 \le p < \Ph - 1$. 
The  canceling of opinions 
in this phase does not change the difference between the number of opinions $A$ and $B$.
If each node entering the next phase remains in a normal state, then for each node
entering the next phase and having a vote, this vote must be a result of (exactly) one duplication
in the doubling stage of phase $p$. 
Thus, if there is a (global) step when all nodes are
in the beginning part of phase $p+1$ and all are in normal states,
then the number of votes is $2\cdot 2^p|a_0 - b_0|$.

Lemma~\ref{Lem:doublingVoteCount} implies the following corollary.

\begin{corollary}\label{Cor:time}
If for each $p = 0, 1, \ldots, P-1$, there is a (global) step such that all nodes are in the beginning of phase $p$,
then for some phase $q \le \log (n/|a_0 - b_0|) + 1$, 
there is a node which is in the beginning of this phase and in either "$\done$" or "$\fail$" state. 
\end{corollary}

%

{\bf Proof of Theorem \ref{Thm:ourExactMajorityProtocol}:}

Use Lemma~\ref{Lem:phase} and Corollary~\ref{Cor:time} to show,
by induction on the number of phases, that by the phase $\log (n/|a_0 - b_0|) + 1$,
\whp\ at least one node reaches a $\done$ state with the correct majority opinion, no minority opinions 
are left in the system and no node is in a $\fail$ state. 
The correct answer will be broadcast to all nodes in additional $O(\log n)$ time \whp

To obtain $O(\log^2 n)$ time also in expectation, we have to augment our protocol 
to recover from the possibility of ending up in $\fail$ states with incorrect answer.
We run in parallel our protocol and the four-states exact majority protocol: 
the state space of the combined protocol is the product of the state spaces of both protocols
and on each individual interaction, each protocol performs its own state transition.
The output function of the combined protocol
is the output function of our protocol for all states other than the $\fail$ states and 
is the output function of the four-state protocol on all $\fail$ states.
Our protocol always converges either to a global state such that 
all nodes are in $\done$ states and all hold the opinion of the initial majority, or
to a global state such that 
all nodes are in $\fail$ states. 
This property and the convergence property of the four-state protocol imply that the combine 
protocol converges to the correct answer. 
Theorem~\ref{Thm:ourExactMajorityProtocol} implies that 
the combined protocol computes the exact majority in $O(\log n \log (n/|a_0 - b_0|))$ time \whp: 
if our protocol stabilizes on the $\done$ states and 
the correct output, the four-state protocol will never overrule it.

Our protocol converges in $O(\log^2 n)$ time \whp\
Thus the expected convergence time of the combined protocol is $O(\log n \log (n/|a_0 - b_0|))$ with a high probability $p$
and is $O(n\log n)$ (the expected convergence time of the four-state protocol)
with probability $1-p$.
With the appropriate constants in our protocol, $p \ge 1 - n$, so 
the expected convergence of the combined protocol is $O(\log n \log (n/|a_0 - b_0|))$.

\subsection{State transitions for Section \ref{elect}: first part of the leader election protocol}

In the specification of the transitions below, 
we always specified the transition of the first entry. That is, 
a transition $(q_i,g_j) \rightarrow \, q_l$ means that 
if a node $v$ in state $q_i$ meets a node $u$ in state $q_j$, then $v$ switches to state $q_l$. 
Also, we omit the $X$- and $G$-entries from the state of a node, if they do not change and 
do not play any role in the transition.
If the transition for some specific state by meeting a node in some other specific state is is not specified, then 
this means that the state will not change.

The transition function for the phases $q_0$ and $q_1$ are defined below.

\begin{framed}
\begin{itemize}
\item $(q_0^i,*) \rightarrow \, q_0^{i+1}$, where~$i< c \log n, * \neq q_1$
\item $(q_0^{c \log n},*) \rightarrow \, < q_1^{0},A,N > $, where~$* \neq q_1$
\item $(q_0^i,< q_1^{i'},A, * > ) \rightarrow \, < q_1^0,B,N > $
\item $(q_0^i,< q_1^{i'} ,B,*>) \rightarrow \, <q_1^0,A,N>$
\end{itemize}
\end{framed}

\begin{framed}
\begin{itemize}
\item $(q_1^i,*) \rightarrow \, q_1^{i+1}$, where $i< C/4 \cdot \log n$
\item $(<q_1^{C/4 \cdot \log n},*,*>,<*,N,*>) \rightarrow \, <q_2^0,*,0>$
\item $(<q_1^{C/4 \cdot \log n},*,*>,<*,A,*>) \rightarrow \, <q_2^0,*,0>$
\item $(<q_1^{C/4 \cdot \log n},*,*>,<*,B,*>) \rightarrow \, <q_2^0,*,1>$
\end{itemize}
\end{framed}

The transitions for the phases $q_2$ and $q_4$ are specified in the next table.
\begin{framed}
\begin{itemize}
\item $(q_2^i,*) \rightarrow \, q_2^{i+1}$, where $i \leq C \log n-1$
\item $(q_2^{C \log n},*) \rightarrow \, q_4^{0,0}$
\item $(q_4^{s,t},<*,A,1>) \rightarrow \, q_4^{s+1,t+1}$, if $t \leq C\log n-1$ 
\item $(q_4^{s,t},<*,A,0>) \rightarrow \, q_4^{s,t+1}$, if $t \leq C \log n-1$
\item $(<q_4^{s, C\log n},A,*>,*) \rightarrow \, q_5^{s,0}$
\item $(<q_4^{s,C \log n},B,*>,*) \rightarrow \, q_5^{s+C \log n,0}$
\end{itemize}
\end{framed}

Next we consider the transitions for the nodes in state $q_5$.
\begin{framed}
\begin{itemize}
\item $(q_5^{s,i},*) \rightarrow \, q_5^{s,i+1}$, if $i < 3C \log n$ and $* \neq q_5^{s',i'},q_3^{s',i'}$ with   
$s' >s$
\item $(q_5^{s,i},q_5^{s',i'}) \rightarrow \, q_3^{s',0}$, if $s < s'$
\item $(q_5^{s,i},q_3^{s',i'}) \rightarrow \, q_3^{s',0}$, if $s < s'$
\item $(q_3^{s',i'},q_5^{s,i}) \rightarrow \, q_3^{s,i+1}$, if $s' < s$
\item $(q_3^{s',i'},q_3^{s,i}) \rightarrow \, q_3^{s,i'+1}$, if $s' < s$
\item $(q_3^{s',i'},*) \rightarrow \, q_3^{s',i'+1}$, 
if no other rule for $q_3$ or $q_7^{0,0}$ is specified
\item $(q_3^{s',i'},q_7^{0,*}) \rightarrow \, q_7^{0,0}$ 
\item $(q_3^{s,C^2 \log n},*) \rightarrow \, q_7^{0,0}$
\item $(q_5^{s,3C\log n},*) \rightarrow \, q_6^{0,0,0}$
\end{itemize}
\end{framed}

According to the rules above, if two nodes meet, 
one e.g. in state $q_5^{s,t}$ and the other in $q_3^{s',i}$
with $s>s'$, then 
the first node switches to state $q_5^{s,t+1}$ (according to $(q_5^{s,t},*) \rightarrow \, q_5^{s,t+1}$) 
and the second one to $q_3^{s,i+1}$
(according to $(q_3^{s',i},q_{5}^{s,t}) \rightarrow \, q_3^{s,i+1}$). 
If a $q_3^{s,i}$ node interacts with a $q_5^{s',t}$ node with $s>s'$, then 
the $q_3^{s,i}$ node only increments its counter of interactions,
switching to $q_3^{s,i+1}$.

\subsection{Proofs for Section \ref{AnalysisPartOne}}

{\bf Proof of Lemma \ref{lemAB}:}

We define $n$ consecutive interactions a \emph{period}, and show that if $c$ is large enough, then within 
$c \log n$ periods, every node is chosen at least $c \log n$ and at most $3c \log n$ times for communication, \whp\
Let $v \in V$ be a node. At each time step, $v$ is chosen for communication with probability $2/n$, independently of all 
other steps. Thus, in $c \log n$ periods, $v$ is chosen $2c \log n$ times in expectation. Simple Chernoff bounds 
imply that if $c$ is large enough, then $v$ is chosen at least $c \log n$ times and at most $3c \log n$ times, \whp\

Assume now that $C > 64 c$ is a large constant. 
Similarly as before, if $C$ is large enough, then each node $v$ interacts \whp\
within $(C/8 - C/128) \log n$ periods at least $(C/4 - C/64) \log n$ 
and at most $C/4 \cdot \log n$ times.
Within $(C/8+C/64) \log n$ periods, every node interacts at least $C/4 \cdot \log n$ and at most $(C/4 + C/32) \log n$ times
The pigeonhole principle implies that  
within $C/8 \log n$ periods there will be at least one node which interacts 
at least $C/4 \cdot \log n$ times. Thus, we conclude that \whp\ 
after at most $(C/8 + c) \log n$ periods at least one node will 
switch to state $q_2^{0,0}$.

As described above, by the end of period $c \log n$, all nodes have switched to state $q_1$, but no node reached 
state $q_2^{0,0}$, \whp\ That is, by that time all nodes belong to one of the sets $A$ or $B$, \whp\ Assume now, there 
is a time step such that exactly $n/4$ nodes belong to set $A$, and $|A| > |B|$. This implies that at least 
$n/2$ nodes are still in state $q_0^i$ with some $i \leq c \log n$. Let $v$ be such a node. Whenever $v$ interacts, it meets
an $A$ node with probability at least $1/4$. 

We divide the 
set of $q_0$-nodes into the sets $\{ v_1, \dots , v_{n/4}\}$ and $\{ v_{n/4+1} , \dots , v_{n/2}\}$. 
We assume w.l.o.g.~that 
these nodes are 
ordered according to the time they convert into $q_1$-nodes.  
Then, each node $v_i$ with $i \leq n/4$ 
switches to state $B$ with 
probability at least $1/4$, independently. On the other hand, if two nodes $v_i$ and $v_{i+1}$ meet, where $v_{i+1}$ is in state $q_0^{c \log n}$, then 
both turn into $A$ nodes. Thus, every node $v_i$ with $i \leq n/4$ being odd switches to state $B$ with probability $1/4$, independently. Therefore, 
the expected number of $B$ nodes at the end is at least $n/32$. Applying now 
Chernoff bounds, we obtain the lemma. In the case that at a certain time step $|B| = n/4$ and $|B| > |A|$, the same 
analysis leads to the result.

{\bf Proof of Lemma \ref{lem1}:}

In this proof, we first show that at the time step when the first node
or the first two nodes switch to state $q^{0,0}_4$, \whp\ all the other nodes are in state $q_2$. Furthermore, 
we show that \whp\ there are $\Theta(n)$ nodes of both types, $0$ and $1$, in set $A$. For the set $B$, the same 
can be shown with the same techniques. 

As shown in the previous proof, within $(C/8+C/64) \log n$ periods every node interacts at least $C/4 \cdot \log n$ times. Thus by the end of 
period $(c+C/8+C/64) \log n$, all nodes are in state $q_2$, 
and no node is in state $q_4$ yet. Together with the fact that 
within $(C/8 - C/128) \log n$ periods the number of interactions a node may have is at most $C/4 \cdot \log n$, \whp\ we conclude that 
the nodes turn into state $q_2^{0}$ after the beginning of period $(C/8 - C/128) \log n$ 
but by the end ofperiod $(c+C/8+C/64) \log n$. 
We also know that by the end of period $c \log n \ll (C/8 - C/128) \log n$ all nodes belong to some set $A$ or $B$ \whp\ A node 
in state $q_1^{C/4 \cdot \log n}$ sets its $G$-value to $0$ if it meets an $A$ node, and to $1$ if it meets a $B$ node. Here we only consider 
the nodes of set $A$. 
Clearly, there is a constant probability $p =|A|/n$ between $0$ and $1$ 
that a node $v$ sets its $G$-value to $0$; however,
the events between two nodes of $A$ are not independent. In expectation, there will be $|A|^2/n = \Theta(n)$ nodes in $A$ with $G$-value $0$. 

Let $A = \{ v_1, \dots , v_{|A|}\}$. Let $X_i$ be $1$ if the $G$-value of $v_i$ becomes $0$ and $0$ otherwise, and let $X= \sum_{i=1}^{|A|} X_i$. Define $Y_0 = \Exp[X]$
and $Y_i = \Exp[X ~|~ X_1, \dots , X_{i-1}]$. Clearly, $Y_0, \dots , Y_{|A|}$ is a Martingale sequence, which satisfies the $2$-Lipschitz condition. 
Applying now the Azuma-Hoeffding inequality, we obtain that the number of nodes with $G$-value $0$ in $A$ is $(1\pm o(1)) |A|^2/n$, \whp\ This 
also implies that the number of nodes with $G$-value $1$ in $A$ is $(1-o(1)) |A|\cdot |B|/n$, \whp\ 
Applying the same calculations to the set $B$, we obtain the result w.r.t.~the $0$- and $1$-nodes in $B$.

{\bf Proof of Lemma \ref{lemprev2}:}

As in the proof of Lemma \ref{lemAB}, we call $n$ consequtive interactions a period. We showed that \whp\ 
all nodes convert between periods $(C/8 - C/128) \log n$ and $(c+C/8+C/64) \log n$ into state $q_2^{0,0}$. 
Using Chernoff bounds as in the proof of Lemma \ref{lem1} we obtain that within additional $(C/2 + C/128) \log n$
periods every node will convert into state $q_4$. Next we show that a node $v$ meets within $(Cn/(2|A|) + C/128) \log n$ 
periods at least $C \log n$ nodes from set $A$, \whp, if $C$ is large enough. 

According to Lemma \ref{lemAB}, $|A| > n(1-o(1))/32$, \whp\ A node is selected for communication in a step with probability 
$2/n$, and if selected, it interacts with an $A$-node with probability $|A|/n$. Thus, the expected number of steps 
$v$ interacts with an $A$-node is 
$$\frac{2|A|n(Cn/(2|A|) + C/128) \log n}{n^2} = (C+C|A|/(64n)) \log n \geq (C+C/(2048)) \log n$$
Applying Chenroff bounds, we obtain that the number of steps $v$ interacts with an $A$-node is at least 
$C \log n$, \whp, if $C$ is large enough. Using simlar techniques, we obtain that within $(Cn/(2|A|) - C/128) \log n$
periods the number of steps a node $v$ interacts with an $A$ is at most $C \log n$, \whp, if $C$ is large enough.

This implies that the first period, in which a node switches to state $q_5^{*,0}$ is  
$$(C/8 - C/128) \log n + (C/2 - C/128) \log n+(nC/(2|A|) - C/128) \log n$$ and the last period in which a node converts to state 
$q_5^{*,0}$ is $$(c+C/8+C/64) \log n + (C/2 + C/128) \log n + (Cn/(2|A|) + C/128) \log n ,$$ \whp\ The time difference between them is 
at most $(c+C/32+C/64+C/128) \log n$ periods. Applying Chernoff bounds, we obtain that within these periods no node 
can interact with more than $C \log n$ other nodes, \whp\ Thus, when a node reaches state $q_5^{*,C\log n}$, then 
all $q_4$-nodes have switched to state $q_5$ (and may have turned afterwards into $q_3$-nodes).

{\bf Proof of Lemma \ref{lem2}:}

The first statement follows from Lemma \ref{lemprev2}.
To prove the second statement, we only consider nodes from set $B$, and show that among these nodes the maximum value $i$ is unique with constant probability. Taking into 
account that at the end the value $C \log n$ is added to all nodes from $B$ (while the nodes from $A$ leave their values after $C \log n$ 
interactions with nodes from $A$ unchanged), the maximum value occurs at a node of set $B$, \whp\

The proof of this lemma is similar to the analysis of the number of nodes having maximum degree in an 
Erd{\H{o}}s-R\'enyi random graph with edge probability $\Theta((\log n)/n)$. We first compute the probability that 
a node has an arbitrary but fixed value $i$, when it switches to state $q_5$. Note that every node has 
$C \log n$ interactions with $A$-nodes before it switches to state $q_5^{0,0}$. Let $p$ denote the probability that a node increases its 
$i$ value in a step (i.e., it meets a node in some state $<*,A,1>$). Then, every node in $q_4^{*,*}$ has in expectation an $i$ value of $pC \log n$ when it enters 
state $q_6$. Using simple Chernoff bounds, we obtain that no node will have an $i$ value larger than 
$pC \log n + \delta C \log n$, \whp, where $\delta$ is significantly smaller than $p$ whenever $C$ is large enough. 
This also implies that if $C$ is large enough, then $p+\delta < 1$. 

Consider some value $0 < \delta' < \delta$. Then, denote by $k_{\delta'} = pC \log n + \delta' C \log n$ and 
we obtain that a node is in state $q_4^{k_{\delta'},C \log n}$ before entering state $q_5^{0,0}$ with probability
$$P_{k_{\delta'}} = \left({C \log n} \atop {k_{\delta'}} \right) p^{k_{\delta'}} (1-p)^{c \log n - k_{\delta'}} .$$
We consider now $P_{k_{\delta'}}/P_{k_{\delta'} +1}$
\begin{eqnarray*}
P_{k_{\delta'}}/P_{k_{\delta'} +1} &=& \frac{\left({C \log n} \atop {k_{\delta'}} \right) p^{k_{\delta'}} (1-p)^{C \log n 
- k_{\delta'}}}{\left({C \log n} \atop {k_{\delta'}+1} \right) p^{k_{\delta'}+1} (1-p)^{C \log n - k_{\delta'}-1}} \\
&=& \frac{(k_{\delta'}+1)(1-p)}{(C\log n - k_{\delta'})p} \\
&=& \frac{(p+\delta'+1/(C\log n))(1-p)}{(1-p-\delta')p}
\end{eqnarray*}
which is a constant larger than $1$. This implies that there is some value $\delta'$ for which $P_{k_{\delta'}} \geq 1/n$, 
$P_{k_{\delta'}+1} < 1/n$, and $P_{k_{\delta'}}/P_{k_{\delta'} +1}$ is bounded from below by some constant larger than 1. 
Furthermore, $P_{k_{\delta''}}/P_{k_{\delta''} +1}$ 
is a constant larger than one for any $\delta' \leq \delta'' \leq \delta$. Since among the nodes in $B$ the random variables 
denoting the number of $s$-values are independent, 
there will be at least one node in $B$, which reaches state 
$q_4^{i,C \log n}$ with $s \geq k_{\delta'}$, 
with some probability bounded away from $0$. Taking into consideration that a node exceeds value 
$k_{\delta'}$ with probability 
$$\sum_{i= k_{\delta'}+1}^{k_{\delta}} P_i + o(1/n) = \Theta(1/n) ,$$
this value $s$ is unique among all nodes with constant
probability.
Furthermore, there is a range $[k_{\delta'}-r\log \log n),k_{\delta'}]$ with $r$ being 
some constant, such that 
the probability for a node to reach state $q_4^{s,C \log n}$ with $s$ in the range given before is $\Theta(\log n/n)$. This 
implies that the largest $s$ value will be at least $k_{\delta'} - r\log \log n$, with high probability.

{\bf Proof of Lemma \ref{lem3}:}

According to Lemma \ref{lemprev2}, we know that at the first time step when a node is in state $q_5^{s,C \log n}$, all nodes are 
in state $q_3$ or $q_5$ \whp\ Since the nodes $q_5^{s_{\max},*}$ can not switch to state $q_3$, all these nodes will 
broadcast their values to the other nodes by letting every other node switch to state $q_3^{s_{\max}}$. Thus, the 
spread of the $s_{max}$ values can be modeled by an asynchronous push-pull broadcast protocol as described in 
Section \ref{pushpull}. According to Theorem \ref{the:pp1}, there is a constant $c'$ such that within 
$c' \log n$ periods, all nodes receive the value $s_{\max}$, \whp\ Letting $C$ being large enough, Chernoff bounds imply that 
no node interacts within $c' \log n$ steps more than $C \log n$ times, \whp\ Since at the first time when a node 
is in state $q_5^{s,C \log n}$, all nodes are in state $q_5$ or $q_3$ (including the nodes with upper index $s_{\max}$), and 
within $c' \log n$ further periods all nodes are either in $V_{\max}$ or in state $q_3$, we obtain that by that time no node 
will have switched to state $q_6$ or $q_7$, \whp\ Hence, only nodes of $V_{\max}$ can enter state $q_6$, and when this happens
for the first time, all other nodes are in state $q_3$, \whp\

\subsection{Transition table from Section \ref{sec3}}

We assume that there are $O(\log n)$ nodes which are in state $q_6$ (leader candidates) and the rest is in state $q_7$
(followers, starting with $q_7^{0,0}$) or $q_5$ (leader candidates which have not entered state $q_6$ yet). 
Clearly, at the beginning there can also be 
$q_3$-nodes; however, these followers are treated as $q_7^{0,0}$-nodes, i.e., in each transition $q_7^{0,0}$ on the left 
hand side can be replaced by any $q_3^{i,j}$.
We check whether a single node is in state $q_6$ (and no node is in state $q_5$) by the following algorithm. 
The different states a leader candidate
can have are $q_6^{l,i,j}$ with $l=0,1,2$, and $i,j \leq C \log n$. Similarly, a follower can be in a state $q_7^{l,i}$ with 
$l = 0,1,2$ and $i \leq 3C \log n/4$. 
The transitions for initiating a message of type $1$ are as follows:

\begin{framed}
\begin{itemize}
\item $(q_6^{0,0,j} , <*,A,*>) \rightarrow \, q_6^{1,1,j}$, if $j < C \log n$
\item $(q_6^{1,1,j},q_7^{0,i}) \rightarrow \, q_6^{1,2,j}$
\item $(q_7^{0,i},q_6^{1,1,j}) \rightarrow \, q_7^{1,1}$
\item $(q_6^{1,1,j},q_7^{l,i}) \rightarrow \, q_8$, if $l \neq 0$
\item $(q_6^{1,i,j}, *) \rightarrow \, q_6^{1,i+1,j}$, if $* \neq q_6^{l,i',j'},q_5^{s,i'},q_7^{2,i'}$ and $i,j \neq C \log n$
\item $(q_6^{1,i,j}, *) \rightarrow \, q_8$, if $* = q_6^{l,i',j'},q_5^{s,j'},q_7^{2,i'}$
\item $(q_6^{l,C \log n,j},*) \rightarrow \, q_6^{0,0,j+1}$, if $j < C \log n$
\item $(q_6^{l,C \log n,C \log n},*) \rightarrow \, <q_9,0>$
\end{itemize}
\end{framed}

The transitions for initiating a message of type $2$ are described below.

\begin{framed}
\begin{itemize}
\item $(q_6^{0,0,j}, <*, B,*>) \rightarrow \, q_6^{2,1,j}$, if $j < C \log n$
\item $(q_6^{2,1,j},q_7^{0,i}) \rightarrow \, q_6^{2,2,j}$
\item $(q_7^{0,i},q_6^{2,1,j}) \rightarrow \, q_7^{2,1}$
\item $(q_6^{2,1,j},q_7^{l,i}) \rightarrow \, q_8$, if $l \neq 0$
\item $(q_6^{2,i,j}, *) \rightarrow \, q_6^{2,i+1,j}$, if $* \neq q_6^{l,i',j'},q_5^{s,i'},q_7^{1,i'}$ and $i,j \neq C \log n$ 
\item $(q_6^{2,i,j}, *) \rightarrow \, q_8$, if $* = q_6^{l,i',j'},q_5^{i'',j''},q_7^{1,i'''}$
\end{itemize}
\end{framed}
The transitions above describe the case, in which the Bernoulli trial was unsuccessful.

Next we describe how the message is propagated in the system. 
\begin{framed}
\begin{itemize}
\item $(q_7^{l,i},q_7^{0,i'}) \rightarrow \, q_7^{l,i+1}$ if $0<i<C \log n/32$ and $l =1,2$
\item $(q_7^{0,i'},q_7^{l,i}) \rightarrow \, q_7^{l,1}$ if $0<i<C \log n/32$ and $l=1,2$
\item $(q_7^{l,i},q_7^{l',i'}) \rightarrow \, q_8$ if $l,l' \in \{1,2\}$ and $l \neq l'$
\item $(q_7^{l,i},q_7^{l',i'}) \rightarrow \, q_7^{l,i+1}$ if $1 \leq i < 3C \log n/4$, where $l' =0,l$
\item $(q_7^{l,i},q_5^{s,i'}) \rightarrow \, q_8$ if $l=1,2$
\item $(q_7^{l,3C\log n/4},*) \rightarrow \, q_7^{0,0}$
\item $(q_7^{0,C^3 \log^2 n},*) \rightarrow \, q_8$
\item $(q_7^{0,i},*) \rightarrow \, q_7^{0,i+1}$ if none of the other rules apply
\end{itemize}
\end{framed}

As in Section \ref{sec:alg}, we specified for each $q_i$-node the transition when it meets a node in state $q_j$. 
That is, transition $(q_i,q_j) \rightarrow \, q_{l}$ 
means that the $q_i$-node switches to state $q_l$ when it meets a node in state $q_j$. 
Now we are looking at the endgame.
\begin{framed}
\begin{itemize}
\item $(*,q_8) \rightarrow \, (<q_0^0,N,N>,<q_0^0,N,N>)$ if $* \neq q_9,q_{10}$
\item $(*,<q_0^i,N,N>) \rightarrow \, (<q_0^0,N,N>,<q_0^{i+1},N,N>)$ if $* \neq q_9,q_{10},q_0,q_1$
\item $(<q_9,0>,<q_9,1>) \rightarrow \, (<q_9,1>,q_{10})$
\item $(*,<q_9,l>) \rightarrow \, (q_{10},<q_9,1-l>)$ if $* \neq <q_9,0> \mbox{ or } <q_9,l>$
\item $(*,q_{10}) \rightarrow \, (q_{10},q_{10})$ if $* \neq q_9$
\end{itemize}
\end{framed}


\subsection{Proofs for Section \ref{sec3}}

\begin{lemma}
\label{lem4}
If there is only one leader candidate, then no node will ever enter state $q_8$ \whp  
\end{lemma}
\begin{proof}
In this lemma, we consider the case when $|V_{\max}| =1$ in Lemma \ref{lem2}. According to Lemma \ref{lem3}, 
the node $v \in V_{\max}$ switches to state $q_6^{0,0,0}$, and at this time, all other nodes are in state 
$q_3^{i_{\max},*}$. Since in the testing phase $q_3$-nodes follow the same rules as $q_7^{0,0}$-nodes, we may assume that  
at the beginning, one node is in state $q_6^{0,0,0}$ and all the others are in state $q_7^{0,0}$. Applying simple Chernoff bounds, 
we conclude that with high probability the second upper index of a $q_7^{0,*}$-node will never reach the value $C^3 \log^2 n$ before a $q_8$- or 
$q_9$-node appears. Therefore, we ignore the second upper index in a $q_7$-node as long as the first upper index is $0$.

Let us define the $i^{\mbox{th}}$ epoch as the time during which the third upper index of the $q_6$-node is $i$. We show by induction that 
\whp \ at the beginning of an epoch all nodes except the $q_6$-node are in state $q_7^{0,*}$, during the epoch no node switches to state $q_8$, and at the 
end of the epoch all nodes except the $q_6$-node are in state $q_7^{0,*}$. 

At the beginning of the $0^{\mbox{th}}$ epoch the statement is true (see above). Assume that the statement is true at the end of epoch $i-1$ and therefore
it is also true at the beginning of epoch $i$.
When the node $q_6$ communicates for the first time in this epoch, 
it switches either to state $q_6^{1,1,i}$ or to state $q_6^{2,1,i}$. 
Assume w.l.o.g.~that the node 
switches at some point in time to state $q_6^{1,1,i}$ and all other nodes are in state $q_7^{0,*}$. Then 
the leader candidate interacts with another node exactly once, and converts this node into state $q_7^{1,1}$. 
The process can be described by a simple (asynchronous) push \& pull algorithm \cite{KSSV00}, where in each step an edge is chosen, and if one 
of the nodes possesses the message, then the other one will become informed as well. 

Let $n$ consecutive iterations be a period. 
According to Theorem \ref{the:pp1}, within $O(\log n)$ periods all nodes are informed, \whp\ Thus, if $C$ is large enough, 
then all nodes are informed within $C \log n/70$ 
periods \whp, which implies that at the end of the broadcast process all followers are in state $q_7^{1,j}$ for some $j$. 
Furthermore, if $C$ is large enough, then within $C \log n/32$ periods, all nodes interacted at least once \emph{after they received 
the message, i.e., after they switched to state $q_7^{1,1}$}, \whp\ Also,
within $(3C/8 + C/16) \log n$ periods every node interacts at least $3C/4 \cdot \log n $ times and at most $(3C/4 + C/8) \cdot \log n $
times, \whp\ Thus, after period $(3C/8 + C/16 + C/32) \log n$ all nodes will be in state $q_7^{0,*}$, except the leader candidate that will still be in state 
$q_6^{1,j,i}$ with $j \leq (3C/4 + C/4) \cdot \log n $, \whp\ This implies that when the leader candidate reaches state $q_6^{1,C \log n,i}$ all followers 
are in state $q_7^{0,*}$ and the next epoch $i+1$ starts. Summarizing, at the beginning of the $i^{\mbox{th}}$ epoch 
all nodes except the $q_6$-node are \whp\ in state $q_7^{0,*}$, during this epoch no node switches to state $q_8$, and at the 
end of the epoch all nodes except the $q_6$-node are in state $q_7^{0,*}$.
If $i = C \log n$, then the leader candidate switches to state $q_9$ \whp\ 
and declares itself the leader. 
\end{proof}

In the next lemma we assume that $|V_{\max}|>1$ in Lemma \ref{lem2}, and show that all leader candidates act synchronously within a relatively small time frame or 
a $q_8$-node appears. 
We assume here that 
there are at most $O(\log n)$ leader candidates (this holds with high probability, as shown in Lemma \ref{lem2}).
Lemma~\ref{lem6}
shows that if there are at least two leader candidates, then one of them will realize this 
and will enter state $q_8$.
The node in state $q_8$ informs all nodes and the whole process is repeated.

\begin{lemma}
\label{lemsynch}
Let $n$ consecutive interactions be a period, and assume period $1$ starts when a single node $v$ is in state $q_6^{0,0,0}$ and all 
others are in state $q_5^{i_{\max},*}$ or $q_3^{i_{\max},*}$. 
We call period $\tau_i$ the period in which $v$ performs its $i(C \log n+1)$'st interaction after the beginning of period $1$. 
For any $i < C \log n$, if there is more than 
one leader candidate, then either all leader candidates switch to state $q_6^{0,0,i}$ in one of the  
periods $\{ \tau_i -C \log n/70, \dots , \tau_i +C \log n/70\}$, or a node switches to state $q_8$ before period $\tau_{i+1}$ \whp\
\end{lemma}
\begin{proof}
%
First assume that at the beginning of period $1$, there are some $q_5$-nodes. Clearly, these nodes are in some states 
$q_5^{i_{\max},j}$.
Then, assume w.l.o.g.~that $v$ switches to state $q_6^{1,1,0}$ and initiates the transition of $q_7^{0,*}$-nodes to 
$q_7^{1,*}$. According to Theorem \ref{the:pp1}, within $O(\log n)$ periods a message is spread to all nodes, \whp\ Thus, if $C$ is large enough, 
then all nodes are either converted to $q_7^{1,*}$ within $C \log n/70$ 
periods \whp, or a node switches to state $q_8$. Thus, if one of the $q_5$-nodes does not switch to $q_6^{0,0,0}$ within 
 $C \log n/70$ periods, then this node will switch to state $q_8$, \whp

Now we show the statement on induction over $\tau_i$. Assume that w.r.t. $\tau_{i}$ all leader candidates 
switched to state $q_6^{0,0,i}$ in one of the  
periods $\{ \tau_{i} -C \log n/70, \dots , \tau_{i} +C \log n/70\}$. We can assume here that we are given some predefined 
steps at which the nodes interact and change their state to $q_6^{0,0,i}$. Furthermore, for the analysis in this proof we assume a modified protocol, in which  
as long as none of the $q_6$-nodes 
change their third upper index to $i+1$, no node is allowed to enter state $q_8$, and as soon as one of the $q_6$-nodes switches to state $q_6^{0,0,i+1}$, all $q_7$-nodes 
change their state to $q_7^{0,0}$.
Using simple Chernoff bounds, one can show that some 
node $w$ that switched to $q_6^{0,0,i}$ in some period $\tau'_i$ changes to state $q_6^{0,0,i+1}$ in a period 
lying in the range $[\tau'_i + C \log n - C/70 \cdot \log n, \tau'_i + C \log n + C/70 \cdot \log n]$, \whp , if $C$ is large enough. Since we assumed that 
$\tau'_i \in \{ \tau_{i} -C \log n/70, \dots , \tau_{i} +C \log n/70\}$, we obtain that $w$ switches to state $q_6^{0,0,i+1}$ in some period in the 
range $[\tau'_i + C \log n - C/35 \cdot \log n, \tau'_i + C \log n + C/35 \cdot \log n]$ w.h.p. 

Let $u$ be the first node that switches to state $q_6^{0,0,i+1}$ in a period in the range 
$[\tau_i + C \log n - C/35 \cdot \log n, \tau_i + C \log n + C/35 \cdot \log n]$. Let $\tau'_{i+1}$ denote this period. We know that within
$c' \log n$ periods, all nodes receive the message of $u$, \whp\ (see Theorem \ref{the:pp1}), where $c' < C/70$ if $C$ is large enough. 
That is, all $q_7$-nodes will receive the message of $u$ within $c' \log n$ periods, but none of them will interact more than $C/2 \log n$ times 
by period $\tau_i + C \log n + C/35 \cdot \log n$ after recieving this message, \whp\
However, 
if $v$ switches to state $q_6^{0,0,i+1}$ in some period after $\tau'_{i+1}+ c' \log n$, but before period $\tau_i + C \log n + C/35 \cdot \log n$ (the latter 
event happens \whp), then 
$v$ will switch to state $q_8$. The case, in which a node $u$ exists such that $u$ switches to state $q_6^{0,0,i+1}$ in some period in the range 
$[\tau_{i+1}, \tau_i + C \log n + C/35 \cdot \log n]$, is analyzed in the same way. Thus, the lemma follows.  
\end{proof}

\begin{lemma}
\label{lem6}
If there are at least two leader candidates, then at least one node will switch to state $q_8$, \whp\
\end{lemma}
\begin{proof}
%
In the previous lemma we proved that either all leader candidates act synchronously (within a frame of 
$[-C/70 \cdot \log n,C/70 \cdot \log n]$ periods, or a node switches to state $q_8$, \whp\ Assume, there is 
an $i < C \log n$ such that a node $v$ switches to state $q_6^{1,1,i}$ and another node $u$ switches to state 
$q_6^{2,1,i}$. If the two do not switch within the given time frame, then a $q_8$-node arises \whp\ Otherwise, both start 
broadcasting their messages ($v$ a message of type $1$, and $u$ a message of type $2$). However, since a message is 
distributed in the system within $O(\log n)$ steps (see Theorem \ref{the:pp1}), the two messages will meet within 
$C/70 \cdot \log n$ periods, and a $q_8$ node arises. Clearly, the event that for an arbitrary but fixed $i < C \log n$ 
a leader candidate $v$ switches to state $q_6^{1,1,i}$ and another leader candidate $u$ switches to state $q_6^{2,1,i}$
occurs with constant probability. Since we consider any $i = 1, \dots , C\log n$, such an $i$ exists with high probability.
\end{proof}


Theorem \ref{theo1} follows from the previous lemmas.

{\bf Proof of Theorem \ref{theo1}:}

Assume, there is a single leader elected during this process. Then, this leader will be in state $q_9$ at the end
(see Lemma \ref{lem4}), and this information is propagated to all nodes in the system. If there are two or more leader 
candidates, then a node will switch \whp\ to state $q_8$, and this information is also spread to all nodes. For this, we consider the 
transitions $(q_8,*) \rightarrow \, (q_0^{0},q_0^{0})$ and $(q_0^i,*) \rightarrow \, (q_0^{i+1},q_0^{0})$, where 
$* \neq q_0,q_1$. Thus, choosing $c$ accordingly, within $c/4 \cdot \log n$ periods all nodes convert to $q_0$ 
\whp, and no node performs more than 
$c \log n$ interactions, so we are in the initial state, \whp\ Then, with constant probability a leader is chosen. By 
repeating the whole protocol $\Theta(\log n)$ times, a leader is chosen with high probability.  

Concerning the running time, assume that the process is repeated $c' \log n$ times, i.e., $c' \log n-1$ times a $q_8$ node arises that 
restarts the whole protocol. 
Define $Z_i$ to be the random variable, 
which describes how often the testing phase is repeated 
between the $i$'th and $i+1$'st restart of the whole process, 
i.e., the time the third upper index of a fixed leader candidate 
is increased in $q_6^{*,*,j}$ within 
this repetition. According to Lemma \ref{lem6}, with some constant probability $p>0$, a node $q_8$ is created within one execution of the 
testing phase (between states $q_6^{*,*,j}$ and $q_6^{*,*,j+1}$ of the same leader candidate for some $j$). Thus, $Z_i$ can be modeled 
by a random geometric variable 
with $\Pr[Z_i =k] = (1-p)^{k-1}p$. Set $Z = \sum_{i=1}^{c' \log n} Z_i$. Then, applying Chernoff bounds to independent random 
geometric variables \cite{AD11} we have
$$\Pr[Z \geq C' \log n] \leq \exp\left(-\frac{(C'-c')^2 c' \log n}{2(1+C'-c')}\right) < n^{-5}$$
for any $C'$ large enough. Thus, one execution of the testing phase is repeated $O(\log n)$ times \whp\ leading to a total number of 
$O(\log^2 n)$ periods. At the end, when a single leader emerges the total number of periods is also $O(\log^2 n)$ \whp\ (leader election 
requires $O(\log n)$ periods while the overall testing requires $O(\log^2 n)$ periods), and the theorem follows.

{\bf Proof of Theorem \ref{Thm:LreaderElectionExpectedLog2}:}

To obtain a leader-election protocol which converges in $O(\log^2 n)$ time both \whp\ and in expectation,
we follow the same general idea as in the proof of Theorem~\ref{Thm:ourExactMajorityProtocol} for
the exact-majority problem.
We combine a fast protocol which converges in poly-logarithmic time \whp\
(but with some low probability may require long time or might even not converge at all) with
a slow protocol which converges in polynomial time but both \whp\ and in expectation. 
Such combining of fast and slow protocols turns out to be more complicated for our leader election protocol than it was 
for the exact majority protocol because the nodes in the leader election protocol do not have spare states 
to keep their overall count of interactions. This means that they do not know when they 
could declare the failure of the fast protocol and switch to relying on the slow protocol.

Our solution to this problem is to add a pre-processing stage to the whole algorithm.
All nodes start in the same state and they first split into to two groups, say $\a$ nodes and $\b$ nodes, 
using an interaction process analogous to the initial phase of our leader election protocol 
when the nodes decide between the sets $A$ and $B$.
Thus, \whp\ all node decide whether they belong to the $\a$ or the $\b$ group and the size of each group
is $\Theta(n)$. 
While with some exponentially small probability the $\b$ group may be empty, the $\a$ group will always 
have at least one node (as it was the case with the $B$ and $A$ sets in our main protocol).

Each node ends the pre-processing stage after counting $c'\log n$ interactions.
The $\b$ nodes switch then to the state $<q^0,N,N>$ of our fast leader election protocol and follow this protocol, 
disregarding the interactions with $\a$ nodes.
The constant $c'$ is such that \whp\ there is a step when each node has decided on the $\a$ or $\b$ group,
there are $\Theta(n)$ $\b$ nodes, and all these nodes are in the states $<q^i,N,N>$, with $0 \le i \c log n$, 
of our fast protocol (the starting configuration for this protocol).
Since \whp\ we have $\Theta(n)$ $\b$ nodes, it can be shown that the \whp\ $O(\log^2 n)$ bound on the completion time 
of the fast protocol still holds: the progress of the computation will be slower, but \whp\ only by a constant factor.

The $\a$ nodes follow the slow polynomial time protocol and keep counting the number of interaction up to the limit 
of $K\log^2 n$, where the constant $K$ is sufficiently large to guarantee that \whp\ 
by the step when any $\a$ node reaches this limit of interaction,
the $\b$ nodes will have reached the 
final configuration of their fast protocol (one $\b$ node in the leader state $q_9$, while all others in the follower state $q_10$).
In the low probability event that some $\a$ node reaches the limit of $K\log^2 n$ interaction without meeting an $\a$ node
in a final state $q_9$ or $q_10$, this node initiates spreading information throughout the system 
that all nodes have to accept the leader-election outcome of the slow protocol.

The combine protocol will have the $O(\log^2 n)$ \whp\ time of the fast protocol.
Its expected completion time will be $1 - n^{-b}$ times $O(\log^2 n)$ (for the case when the fast protocol is successful)
plus $n^{-b}$ time $O(n^{a})$ (for the case when the fast protocol fails), so also $O(\log^2 n)$
since we can take constants $b > a$. ($O(n^{a})$ is the polynomial expected completion time of the slow protocol).

\subsection{Push-Pull broadcast protocol in the asynchronous model}
\label{pushpull}

We are given a complete graph of $n$ nodes, and one of the nodes possesses a piece of information. 
In each step, two nodes are selected uniformly at random. If one of the nodes has the message, then 
at the end of the step the other node has the message too. We call $n$ consecutive steps a period, and show that
after $O(\log n)$ periods all nodes have the message. The proof is very similar to the one of the 
synchronous model~\cite{KSSV00}. First we show the following lemma.

\begin{lemma}
\label{lem:pp1}
There is a constant $c$ such that after $2c \log n$ periods at least $c \log n$ nodes have the message, \whp 
\end{lemma}
\begin{proof}
Let $v$ be the node which is informed at the beginning. In a period, this node is not selected with probability 
$(1-2/n)^n \approx 1/e^2$. This implies that $v$ is chosen for communication in $2c \log n$ steps less than
$c \log n$ times with probability at most 
$$\sum_{i= c\log n}^{2c \log n} \left( {{2c \log n} \atop i} \right) \frac{1}{e^{2i}} \left(1-\frac{1}{e^2}\right)^{2 c \log n -i}
< \frac{1}{n^2}$$
if $c$ is large enough. Furthermore, $v$ communicates in some step with a node, which has already been chosen for 
communication before, with probability at most $2 c \log n/n$. Thus, node $v$ chooses at least $c \log n$ different nodes 
for communication with probability at least $1-2(c \log n)^2 /n - 1/n^2$.
\end{proof}
Once $c \log n$ nodes are informed, we can show the following lemma.
\begin{lemma}
\label{lem:pp2}
Assume the number of informed nodes $|I(t)|$ at the beginning of period $t$ is at least $c \log n$ but at most 
$n/3$. Then, we have 
$|I(t+1)| \geq \frac{6}{5} |I(t)|$, \whp
\end{lemma}
\begin{proof}
As described above, a node is not chosen for communication within one period with probability at most 
$1/e^2$. We consider now one single period. If there are $|I(t)| > c \log n$ informed nodes at the 
beginning of period $t$, then since a node is chosen less than $\log n$ times for communication \cite{RS98}, 
the method of bounded differences implies
that more than $|I(t)|/5$ informed nodes are not chosen for communication with probability at most 
$$\exp \left(-\frac{|I(t)|^2 (1/5 -1/e^2)^2}{2|I(t)| \log^2 n}\right) < \frac{1}{n^2} .$$
Here we assume for simplicity that $|I(t)| > \log^4 n$, otherwise simple Chernoff bounds can be applied to show the 
statement.
Now we consider one single step, in which such an informed node is chosen. The other node is an informed node 
(from the set $I(t)$ or one that has been informed in some step of this period 
before), with probability at most 
$2|I(t)|/n$. Thus, the expected number of unsuccessful communications in a period is at most 
$8|I(t)|^2/(5 n^2) +|I(t)|/5$, \whp\ Applying Chernoff bounds, we obtain that the number of unsuccessful communications 
is at most $12 |I(t)|/15$, \whp, if $c$ is large enough and $|I(t)| < n/3$. This implies that the number of newly infomred 
nodes is at least $3|I(t)|/15$, \whp, and the lemma follows.  
\end{proof}
Then, we can show that within additional $O(\log n)$ periods, all nodes are informed, \whp
\begin{lemma}
\label{lem:pp3}
If $I(t)| > n/3$, then within additional $O(\log n)$ periods all nodes are informed.
\end{lemma}
\begin{proof}
Let $u$ be an uninformed node. In a period, $u$ is not chosen with probability 
$1/e^2$, and within $2c \log n$ periods, $u$ is chosen in at least $c \log n$ periods, \whp\
(see the proof of Lemma \ref{lem:pp1}.
In each of these periods, $u$ communicates with an informed node with probability at least $1/3$.
Thus, $u$ is not informed in these $c \log n$ periods with probability $1/3^{c \log n} < 1/n^2$, if 
$c$ is large enough. Applying the union bound over all nodes, we obtain the lemma.
\end{proof}
Lemmas \ref{lem:pp1}-\ref{lem:pp3} imply the following theorem.
\begin{theorem}
\label{the:pp1}
In a complete graph, a message is spread to all nodes within $O(\log n)$ periods, \whp
\end{theorem}

\subsection{Simulation Results}

In this section we present our simulation results for both exact majority
and leader election on complete graphs. 
We developed our implementations  in the C++11 programming language. Our tests were
conducted on a 144-core machine 
at University of Salzburg and on 384-core and 1024-core
systems at the University of Linz.

{\bf Exact Majority}

We compared our algorithm to the protocols of
\cite{DBLP:conf/soda/AlistarhAEGR17} and \cite{Mertzios-etal-ICALP2014},
referring to the protocol in~\cite{DBLP:conf/soda/AlistarhAEGR17} as AAEGR
and to the protocol in~\cite{Mertzios-etal-ICALP2014} as MNRS,
while calling our protocol BCER.
For the two 
two constants $c$ and $C$ in our algorithm, we choose $c = 10$ and $C = 200$.

In Figure \ref{fig:majority-rounds}, we compared the number of rounds for all
algorithms. The x-axis shows the size of the graph and the y-axis the
average number of rounds per node.
For each $n$ we performed 30 different simulations, and
the initial majority was set to $A = \lfloor n/2 \rfloor + 1$.

We recall that the proven bounds 
on the expected (parallel) time are $O\left(\log^3 n\right)$  for the AAEGR
algorithm of~\cite{DBLP:conf/soda/AlistarhAEGR17} 
and $O\left(\log^2 n\right)$ for our BCER algorithm.
To compare the shape of the curves, we removed the constants for the actual
runtime and compared the normalized number of rounds for the algorithms (see
Figure \ref{fig:majority-rounds-normalized}).  That is, we normalized each
algorithm individually by computing the number of rounds for each of them at $n
= 1000$, and then we divided the values we obtained for different $n$ by this
base number of rounds.

\begin{figure}
\centering
\includegraphics[width=0.8\textwidth]{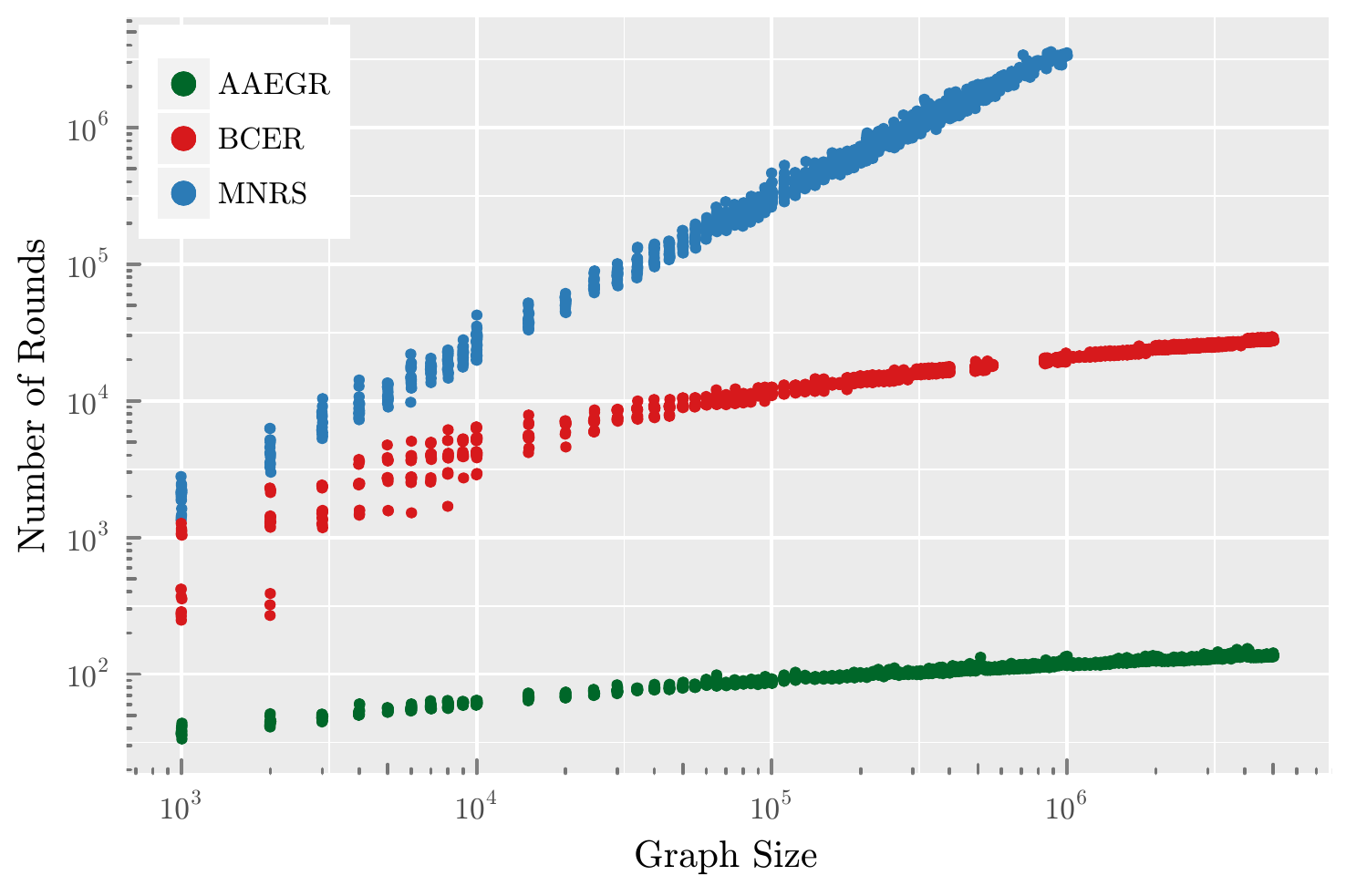}
\caption{Comparison of the number of rounds until exact majority was found. Each dot is a individual test run.}
\label{fig:majority-rounds}
\end{figure}

We observe in Figure~\ref{fig:majority-rounds} 
that the absolute round count for our algorithm is quite high in comparison with
the AAEGR protocol of~\cite{DBLP:conf/soda/AlistarhAEGR17}. 
Figure~\ref{fig:majority-rounds-normalized} suggests that both algorithms
complete the computation within
$\Theta(\log^2 n)$ expected number of rounds, but our algorithm
has a considerably higher constant factor in this bound.
It seems that this high constant factor is the price we pay for
provably guaranteeing the $O(\log^2 n)$ bound 
in expectation and with high probability.

\begin{figure}
\centering
\includegraphics[width=0.8\textwidth]{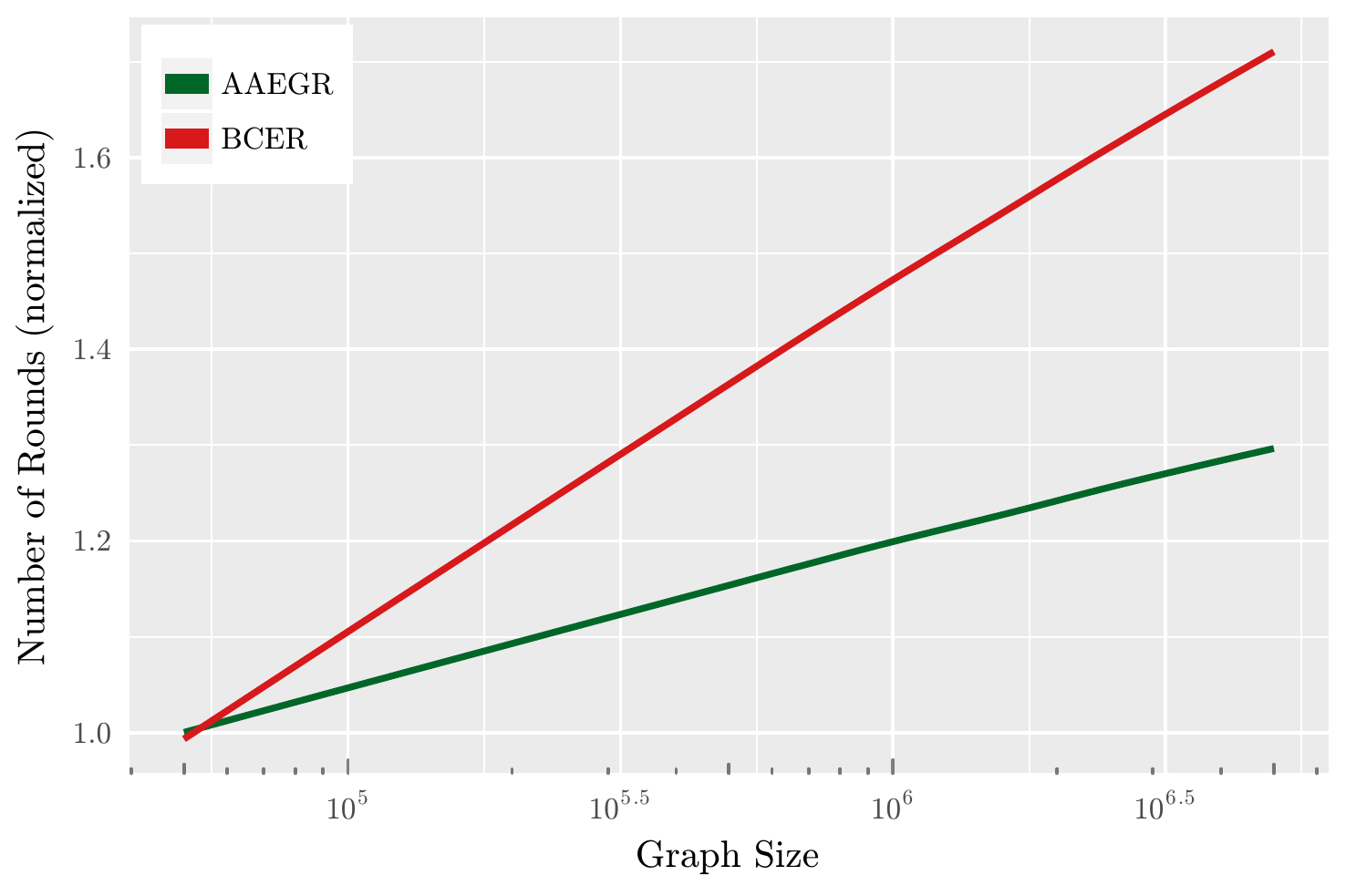}
\caption{Comparison of the normalized number of rounds until exact majority was found. We computed for each data point the square root and calculated the average for each $n$}
\label{fig:majority-rounds-normalized}
\end{figure}

{\bf Leader Election}

We compare our leader election algorithm with 
to the protocols proposed in
\cite{DBLP:conf/icalp/AlistarhG15} (protocol AG) and 
in~\cite{DBLP:conf/soda/AlistarhAEGR17} (protocol AAEGR).
Algorithm \cite{DBLP:conf/icalp/AlistarhG15} uses a state space of
$\Theta\left(\log^3 n\right)$ and each node starts with value 1. In each
round, two nodes compete and compare their values. The winner increases its
value and the looser changes to a minion state and decreases its value.

The algorithm of \cite{DBLP:conf/soda/AlistarhAEGR17} uses $O\left(\log^2
n\right)$ states. Each node can be in one of four modes. All nodes starts with
the same state in a 'seeding mode' where they utilize a syntetic coin flip to
generate randomness. In the subsequent 'lottery mode,' nodes try to
increase their state values. During the tournament mode, nodes compare their
values. The winner stays in competion, the looser changes to the minion mode and
helps the leader candidate to win the competion faster by propagating its value
to the other nodes.

Our algorithm contains several constants, which were optimized empirically.
Each node is in one of 8 phases (states $q_0$ to $q_7$) and the constants have
to be set so that no two nodes are in phases which are too far apart.  For
example, if a node is in state $q_0$ and another one is in state $q_2$, then
they are too far apart.
%
Starting with high constants (e.g.\ 100), we were gradually
decreasing them for each phase seperatly until we found the
smallest constant for each phase, where the property above is still valid. The
results are shown in Table~\ref{tbl:cer-leader-constants}.

\begin{table}
\centering
\begin{tabular}{cc}\toprule
Phase & Constant\\\midrule
$q_0$ & 2\\
$q_1$ & 4\\
$q_2$ & 8\\
$q_4$ & 8\\
$q_5$ & 20\\
$q_6$ & 48, for $2^{\text{nd}}$ index\\
$q_6$ & 4, for $3^{\text{rd}}$ index\\
$q_7$ & 2, for spreading the message\\
$q_7$ & 33, for keeping the message\\
\bottomrule
\end{tabular}
\caption{Constants used for the simulations.}\label{tbl:cer-leader-constants}
\end{table}

First, we compared the number of rounds for these three algorithms (see Figure
\ref{fig:leader-rounds}). Again, the x-axis is the graph size and the y-axis is
the average number of (parallel) rounds.
We show two plots for our algorithm.  The plot BCER is the running time of
our algorithm, where the time is counted until a single $q_9$-node appears.
We also plotted ``BCER w.h.p.'' where we
stopped the algorithm as soon as one single $q_6^{0, 0, 0}$-node appeared
and all other nodes were in the state $q_3$. This configuration leads to a single leader
with high probability according to our theoretical analysis, and has always led to a single 
leader in our simulations.

The algorithm AAEGR of \cite{DBLP:conf/soda/AlistarhAEGR17}, with expected
parallel time $O\left(\log^{5.3} n\right)$ shows a twofold behavior. In the
first fast case, only one node has at the end the highest $phase$-$level$
combination. In this case, the process converges quickly, since the minions
spread this information efficiently through the graph.  In the second slow case,
there are at least two nodes with the highest $phase$ value and none of them
increased its level during the tournament mode. In this case, these two nodes
have to meet in order to let one of them win the tournament. During our
tests, roughly one third of all runs was in the slow case.

\begin{figure}
\centering
\includegraphics[width=0.8\textwidth]{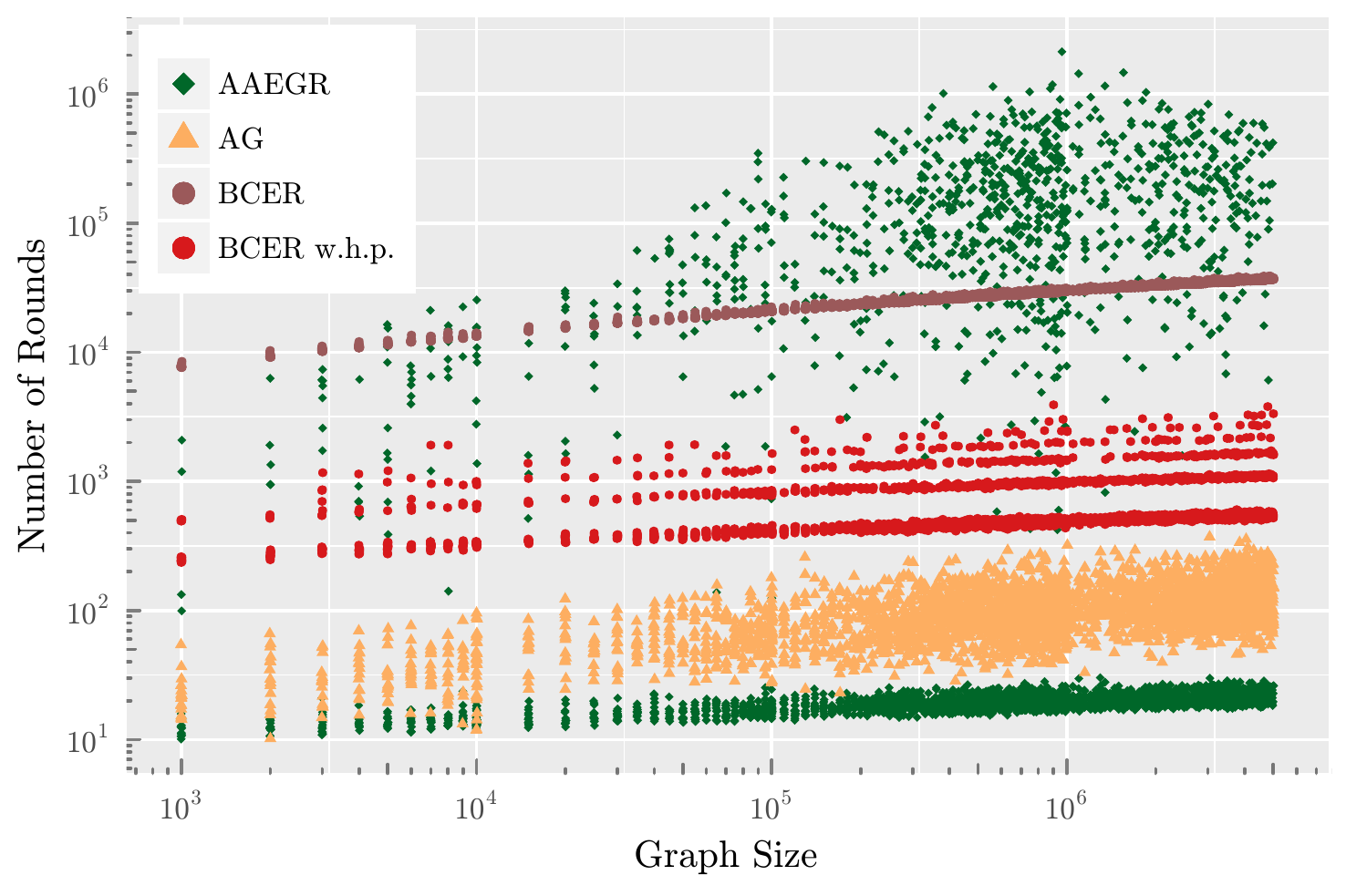}
\caption{Comparison of the number of rounds until a leader was found. Each dot is a individual test run.}
\label{fig:leader-rounds}
\end{figure}

In Figure \ref{fig:leader-rounds-normalized}, we show the normalized number of
rounds for each algorithm. The method to normalize these is the same as
in exact majority case.

\begin{figure}
\centering
\includegraphics[width=0.8\textwidth]{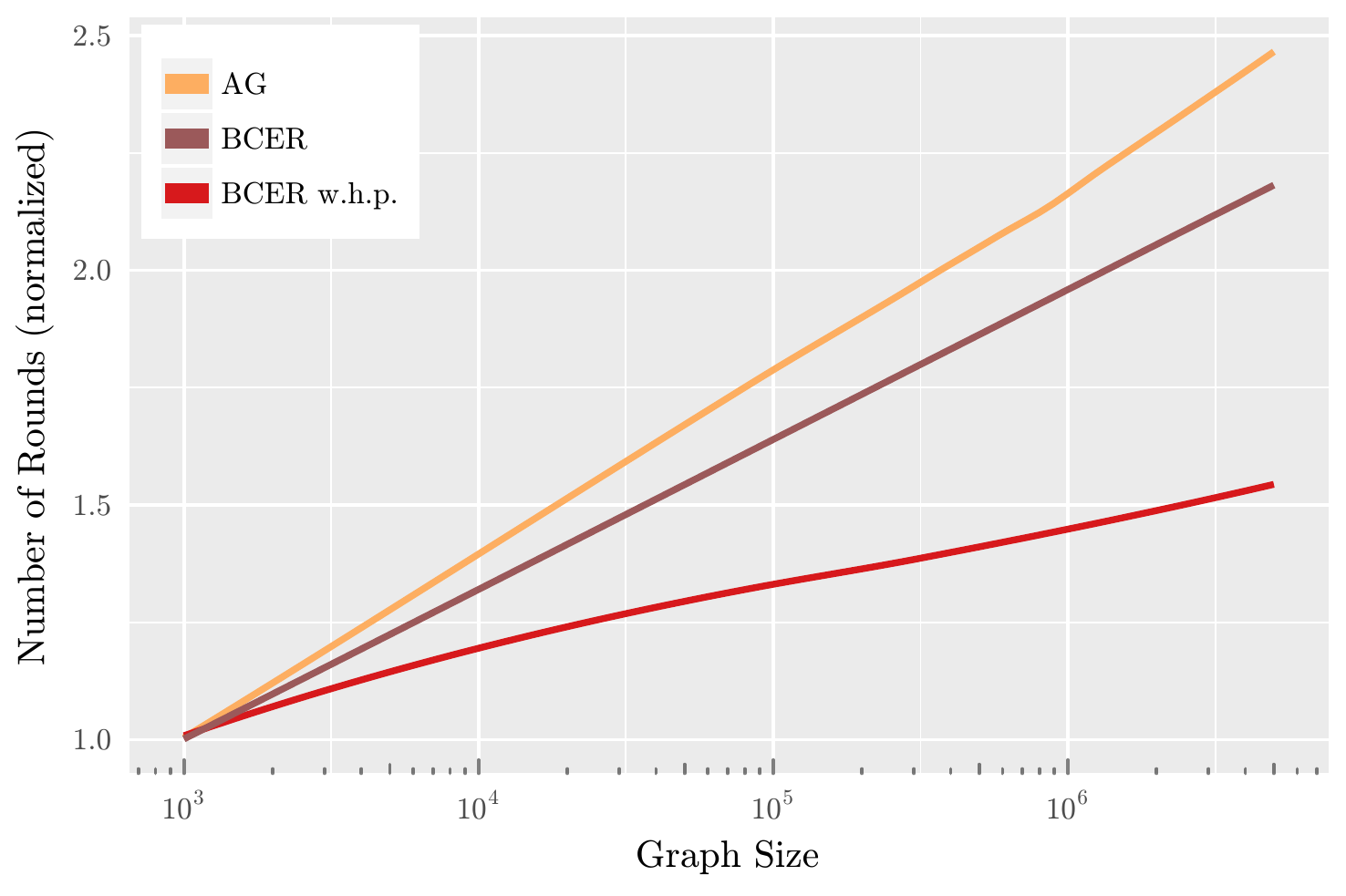}
\caption{Comparison of the normalized number of rounds until a leader was found. We computed for each data point the square root and calculated the average for each $n$}
\label{fig:leader-rounds-normalized}
\end{figure}


\begin{thebibliography}{10}

\bibitem{aldous-fill-2014}
D.~Aldous and J.~A. Fill.
\newblock Reversible {Markov} {Chains} and random walks on graphs, 2002.
\newblock Unfinished monograph, recompiled 2014, available at
  \url{http://www.stat.berkeley.edu/$\sim$aldous/RWG/book.html}.

\bibitem{DBLP:conf/soda/AlistarhAEGR17}
D.~Alistarh, J.~Aspnes, D.~Eisenstat, R.~Gelashvili, and R.~L. Rivest.
\newblock Time-space trade-offs in population protocols.
\newblock In P.~N. Klein, editor, {\em Proceedings of the Twenty-Eighth Annual
  {ACM-SIAM} Symposium on Discrete Algorithms, {SODA} 2017, Barcelona, Spain,
  Hotel Porta Fira, January 16-19}, pages 2560--2579. {SIAM}, 2017.

\bibitem{DBLP:conf/icalp/AlistarhG15}
D.~Alistarh and R.~Gelashvili.
\newblock Polylogarithmic-time leader election in population protocols.
\newblock In M.~M. Halld{\'{o}}rsson, K.~Iwama, N.~Kobayashi, and B.~Speckmann,
  editors, {\em Automata, Languages, and Programming - 42nd International
  Colloquium, {ICALP} 2015, Kyoto, Japan, July 6-10, 2015, Proceedings, Part
  {II}}, volume 9135 of {\em Lecture Notes in Computer Science}, pages
  479--491. Springer, 2015.

\bibitem{DBLP:conf/podc/AlistarhGV15}
D.~Alistarh, R.~Gelashvili, and M.~Vojnovic.
\newblock Fast and exact majority in population protocols.
\newblock In C.~Georgiou and P.~G. Spirakis, editors, {\em Proceedings of the
  2015 {ACM} Symposium on Principles of Distributed Computing, {PODC} 2015,
  Donostia-San Sebasti{\'{a}}n, Spain, July 21 - 23, 2015}, pages 47--56.
  {ACM}, 2015.

\bibitem{DBLP:journals/dc/AngluinADFP06}
D.~Angluin, J.~Aspnes, Z.~Diamadi, M.~J. Fischer, and R.~Peralta.
\newblock Computation in networks of passively mobile finite-state sensors.
\newblock {\em Distributed Computing}, 18(4):235--253, 2006.

\bibitem{AngluinAE2008fast}
D.~Angluin, J.~Aspnes, and D.~Eisenstat.
\newblock Fast computation by population protocols with a leader.
\newblock {\em Distributed Computing}, 21(3):183--199, Sept. 2008.

\bibitem{AD11}
A.~Auger and B.~Doerr.
\newblock {\em Theory of Randomized Search Heuristics: Foundations and Recent
  Developements}.
\newblock World Scientific, 2011.

\bibitem{DBLP:conf/icassp/BenezitTV09}
F.~B{\'{e}}n{\'{e}}zit, P.~Thiran, and M.~Vetterli.
\newblock Interval consensus: From quantized gossip to voting.
\newblock In {\em Proceedings of the {IEEE} International Conference on
  Acoustics, Speech, and Signal Processing, {ICASSP} 2009, 19-24 April 2009,
  Taipei, Taiwan}, pages 3661--3664. {IEEE}, 2009.

\bibitem{DBLP:journals/corr/BerenbrinkFKMW16}
P.~Berenbrink, T.~Friedetzky, P.~Kling, F.~Mallmann{-}Trenn, and C.~Wastell.
\newblock Plurality consensus via shuffling: Lessons learned from load
  balancing.
\newblock {\em CoRR}, abs/1602.01342, 2016.

\bibitem{DBLP:journals/siamdm/CooperEOR13}
C.~Cooper, R.~Els{\"{a}}sser, H.~Ono, and T.~Radzik.
\newblock Coalescing random walks and voting on connected graphs.
\newblock {\em {SIAM} J. Discrete Math.}, 27(4):1748--1758, 2013.

\bibitem{DBLP:conf/wdag/CooperERRS15}
C.~Cooper, R.~Els{\"{a}}sser, T.~Radzik, N.~Rivera, and T.~Shiraga.
\newblock Fast consensus for voting on general expander graphs.
\newblock In {\em Distributed Computing - 29th International Symposium, {DISC}
  2015, Tokyo, Japan, October 7-9, 2015, Proceedings}, pages 248--262, 2015.

\bibitem{DotySoloveichik-DISC2015}
D.~Doty and D.~Soloveichik.
\newblock Stable leader election in population protocols requires linear time.
\newblock In {\em Distributed Computing - 29th International Symposium, {DISC}
  2015, Tokyo, Japan, October 7-9, 2015, Proceedings}, 2015.

\bibitem{DBLP:conf/infocom/DraiefV10}
M.~Draief and M.~Vojnovic.
\newblock Convergence speed of binary interval consensus.
\newblock In {\em {INFOCOM} 2010. 29th {IEEE} International Conference on
  Computer Communications, Joint Conference of the {IEEE} Computer and
  Communications Societies, 15-19 March 2010, San Diego, CA, {USA}}, pages
  1792--1800. {IEEE}, 2010.

\bibitem{KSSV00}
R.~M. Karp, C.~Schindelhauer, S.~Shenker, and B.~V{\"{o}}cking.
\newblock Randomized rumor spreading.
\newblock In {\em 41st Annual Symposium on Foundations of Computer Science,
  {FOCS} 2000, 12-14 November 2000, Redondo Beach, California, {USA}}, pages
  565--574. {IEEE} Computer Society, 2000.

\bibitem{Mertzios-etal-ICALP2014}
G.~B. Mertzios, S.~E. Nikoletseas, C.~L. Raptopoulos, and P.~G. Spirakis.
\newblock Determining majority in networks with local interactions and very
  small local memory.
\newblock In J.~Esparza, P.~Fraigniaud, T.~Husfeldt, and E.~Koutsoupias,
  editors, {\em Automata, Languages, and Programming}, volume 8572 of {\em
  Lecture Notes in Computer Science}, pages 871--882. Springer Berlin
  Heidelberg, 2014.

\bibitem{RS98}
M.~Raab and A.~Steger.
\newblock "balls into bins" - {A} simple and tight analysis.
\newblock In M.~Luby, J.~D.~P. Rolim, and M.~J. Serna, editors, {\em
  Randomization and Approximation Techniques in Computer Science, Second
  International Workshop, RANDOM'98, Barcelona, Spain, October 8-10, 1998,
  Proceedings}, volume 1518 of {\em Lecture Notes in Computer Science}, pages
  159--170. Springer, 1998.

\bibitem{DBLP:conf/focs/SauerwaldS12}
T.~Sauerwald and H.~Sun.
\newblock Tight bounds for randomized load balancing on arbitrary network
  topologies.
\newblock In {\em 53rd Annual {IEEE} Symposium on Foundations of Computer
  Science, {FOCS} 2012, New Brunswick, NJ, USA, October 20-23, 2012}, pages
  341--350. {IEEE} Computer Society, 2012.

\end{thebibliography}
\end{document}